\documentclass{article}

\usepackage{arxiv}

\usepackage[utf8]{inputenc} 
\usepackage[T1]{fontenc}    
\usepackage{hyperref}       
\usepackage{url}            
\usepackage{booktabs}       
\usepackage{amsfonts}       
\usepackage{nicefrac}       
\usepackage{microtype}      
\usepackage{lipsum}		
\usepackage{graphicx}
\usepackage{natbib}
\usepackage{doi}
\usepackage{amsmath,amssymb,amsthm}
 
\usepackage{mathtools}
\usepackage[flushleft]{threeparttable}
\usepackage{longtable}
\usepackage{diagbox}
\usepackage{chngpage}
\usepackage{algorithm}
\usepackage{algpseudocode}

\usepackage{pdflscape}
\usepackage{float}
\usepackage{multirow}
\usepackage{subcaption}
\usepackage{hhline}
\usepackage{xcolor}
\usepackage[english]{babel}

\newtheorem{theorem}{Theorem}
\newcommand{\E}{\mathbb{E}}
\newcommand{\Var}{\mathrm{Var}}
\newtheorem{assumption}{Assumption}
\newcommand{\indep}{\mathrel{\perp\!\!\!\perp}}
\newtheorem{remark}{Remark}
\newtheorem{definition}{Definition}
\newtheorem{lemma}{Lemma}

\newenvironment{proofof}[1]
{\renewcommand{\proofname}{Proof of #1}\begin{proof}}
{\end{proof}}

\title{Estimating heterogeneous treatment effects with survival outcomes via a deep survival learner}

\author{Yuming Sun\\
	Department of Mathematics\\
	William \& Mary, Williamsburg\\
	\texttt{ysun30@wm.edu} \\
	\And
	Jian Kang\\
    Department of Biostatistics\\
	University of Michigan, Ann Arbor\\
	\texttt{jiankang@umich.edu}\\
	\And
	Yi Li\\
	Department of Biostatistics\\
	University of Michigan, Ann Arbor\\
	\texttt{yili@umich.edu}
}

\date{}

\hypersetup{
pdftitle={A template for the arxiv style},
pdfsubject={q-bio.NC, q-bio.QM},
pdfauthor={David S.~Hippocampus, Elias D.~Striatum},
pdfkeywords={First keyword, Second keyword, More},
}

\begin{document}
\maketitle

\begin{abstract}
Estimating heterogeneous treatment effects in survival settings is complicated by right censoring as well as the time-varying nature of the estimand. While the conditional average treatment effect (CATE) provides a natural target, most existing approaches focus on a single prespecified time point and do not account for temporal trajectory, leading to instability in estimation. We propose a deep survival learner (DSL) for estimating heterogeneous treatment effects with right-censored outcomes. The method is based on a doubly robust pseudo-outcome whose conditional expectation identifies time-specific CATEs under standard assumptions. This construction remains unbiased if either the outcome model or the treatment assignment model is correctly specified, when properly accounting for censoring. To estimate CATEs over a clinically relevant time spectrum, DSL employs a multi-output deep neural network with shared representations, enabling joint estimation of treatment effect trajectories. From a theoretical perspective, we derive error bounds for both pointwise and joint estimation over time. We show that joint estimation can leverage temporal structure to control estimation error without incurring much additional approximation cost under smoothness conditions, leading to improved stability relative to separate estimation.
Cross-fitting is incorporated to reduce overfitting and mitigate bias arising from flexible nuisance estimation. Simulation studies demonstrate favorable finite-sample performance, particularly under nuisance model misspecification. Applied to the Boston Lung Cancer Study, DSL reveals  heterogeneity in the effects of perioperative chemotherapy across patient characteristics and over time.
\end{abstract}


\section{Introduction}
\label{intro}
    A substantial literature has examined heterogeneity in treatment effects, often summarized by the conditional average treatment effect (CATE) within the counterfactual framework~\citep{imbens2016causal}. Modeling CATE as a function of pre-treatment covariates enables systematic assessment of variation across patient subgroups, a central goal in precision medicine, where large-scale studies have documented marked heterogeneity in patient characteristics, disease awareness, and clinical outcomes across populations~\citep{dharmarajan2017ckd}.
In survival settings, this task is further complicated by censoring and the time-to-event nature of the outcome. Moreover, biomedical outcomes frequently arise from complex, interdependent processes involving treatment, exposure, and patient characteristics, as illustrated in HPV-related cancers~\citep{haddad2008hpv16} and lung cancer~\citep{zhai2022spirometry}. These features challenge the adequacy of traditional parametric and semiparametric approaches.
{   Recently, metalearner frameworks, including the M-learner, R-learner, and X-learner, have been developed. These approaches leverage the flexibility of modern machine learning methods by decomposing CATE estimation into a sequence of regression problems that can be solved with off-the-shelf algorithms.   For instance, the M-learner constructs pseudo-outcomes based on inverse probability weighting.  The pseudo-outcome formulation is  appealing, as it reduces the problem of estimating CATEs to a regression 
  on covariates~\citep{horvitz1952generalization}.    The X-learner~\citep{kunzel2019metalearners} fits separate outcome models in the treated and control groups and uses them to impute individual treatment effects, while the R-learner~\citep{nie2021quasi} formulates the problem as a residual-on-residual regression. More broadly, doubly robust learners have been developed to improve statistical efficiency and to mitigate sensitivity to nuisance model misspecification~\citep{kennedy2020towards}.}

{ Building on developments for continuous outcomes, recent work has extended metalearner approaches to survival settings. For example, \cite{yang2025doubly} propose a doubly robust learner that combines inverse probability of censoring weighting with nonparametric failure-time Bayesian additive regression trees to construct pseudo-outcomes for CATE estimation. Similarly, \cite{bo2025evaluating} develop alternative pseudo-outcome constructions and extend a range of metalearners, including the X-learner, M-learner, and R-learner, to right-censored data.   Beyond the metalearner framework, tree-based and forest-based methods have also been adapted to estimate heterogeneous treatment effects with survival outcomes. \cite{cui2023estimating} introduce causal survival forests, which extend causal forests to accommodate censoring, while \cite{dandl2024heterogeneous} develop related model-based forest approaches. In a different direction, \cite{hu2024new} propose a random-intercept accelerated failure time model combined with Bayesian additive regression trees to estimate differences in log survival time between treatment groups. Finally, \cite{curth2021survite} introduce SurvITE, a neural network approach that models treatment-specific discrete-time hazard functions to capture treatment effect heterogeneity.}

{ Several methodological challenges remain. Many pseudo-outcome based approaches rely on complete-case information and do not fully account for censoring, which can lead to bias or loss of efficiency in CATE estimation~\citep{xu2023treatment}.   Survival outcomes are inherently time-dependent, and treatment effects may evolve over follow-up due to delayed onset, attenuation, or time-varying biological mechanisms. In many applications, the primary scientific target is therefore the trajectory of heterogeneous treatment effects over a clinically relevant time interval, rather than a contrast at a single prespecified time point.   Most existing methods, however, estimate CATE at a fixed time and treat different time points as unrelated, without leveraging the underlying temporal structure. This can result in unstable estimation and reduced statistical efficiency.}
 
   To address these limitations, we propose a deep survival learner (DSL) for estimating heterogeneous treatment effects with right-censored survival outcomes. We adopt a pseudo-outcome formulation that enables the use of flexible function approximators, such as deep neural networks, to learn its conditional expectation. Our approach distinguishes from existing methods by unifying causal identification, temporal modeling, and robust estimation within a single framework.  First, we introduce a doubly robust pseudo-outcome for survival settings whose conditional expectation identifies the CATE on survival probabilities at each time point. The construction, when properly accounting for right censoring,  is unbiased  if either the outcome model or the treatment assignment model is correctly specified.
Second, we cast CATE estimation as a joint learning problem over time and develop a multi-output neural network to estimate the entire treatment effect trajectory simultaneously. By sharing representations across time points, the method exploits temporal structure in the CATE, enabling information borrowing and leading to smoother and more stable estimates than separate models fit at each time point.
From a theoretical perspective, we  derive error bounds for both pointwise estimation and joint estimation over a time grid. We show that, under mild smoothness conditions on the treatment effect trajectory, joint estimation can control estimation error without incurring much additional approximation cost, thereby improving stability relative to separate estimation. Third, we incorporate cross-fitting to separate nuisance estimation from target estimation, which mitigates overfitting and enhances robustness when more flexible methods are used.

 The remainder of the paper is organized as follows. Section~\ref{methods} introduces the data structure, notation, and target estimand, and presents the proposed DNN estimator. Section~\ref{theory} establishes the error rates at both a fixed time point and over a grid of time points. Section~\ref{simulation} reports simulations  under varying degrees of nuisance model misspecification and different choices of time grids. Section~\ref{BLCS} applies the proposed method to the Boston Lung Cancer Study to examine heterogeneous effects of perioperative chemotherapy among patients with early-stage non-small cell lung cancer.

\section{Methods}
\label{methods}
\subsection{Data Structure,  Estimand and Pseudo-outcome Construction}
\label{notation}
 Let $T$ and $\mathcal C$ denote the event time and censoring time, respectively. With right censoring, we observe $U=\min(T,\mathcal C)$ and the event indicator $\Delta=\mathbb{I}(T\le \mathcal C)$. 
We observe i.i.d. data
\(
\mathcal D=\{(U_i,\Delta_i,X_i,W_i): i=1,\dots,N\},
\)
where $X_i\in\mathbb{R}^d$ denotes pre-treatment covariates and $W_i\in\{0,1\}$ denotes treatment assignment.   Under the potential outcomes framework~\cite{imbens2016causal}, let $T_i^w$ denote the potential survival time under treatment $w\in\{0,1\}$. The observed event time satisfies
\(
T_i = W_i T_i^1 + (1-W_i) T_i^0.
\)
Our target estimand is the conditional average treatment effect (CATE) on the survival probability at  any time $t$ in the support of $U$, defined as
\begin{equation}\label{eq:CATE}
\tau_0(x;t)
=
P(T_i^1>t \mid X_i=x) - P(T_i^0>t \mid X_i=x).
\end{equation}
 
 Since the potential survival times $T_i^0$ and $T_i^1$ in \eqref{eq:CATE} cannot be observed simultaneously (even in the absence of censoring), direct estimation of  CATE  is infeasible. We propose to construct a pseudo-outcome whose  expectation conditional  on $X_i$ equals the target estimand:
\begin{equation}\label{eq:pseudo}
\begin{aligned}
\varphi_i(t) &= \pi_i^1(t)-\pi_i^0(t),  \,\,  {\rm with} \\
\pi_i^w(t)
&=
S^w(X_i;t)
+
\frac{\mathbb I(W_i=w)}{e^w(X_i)}
\left\{
\frac{Y_i(t)}{G^w(X_i;t)}
-
S^w(X_i;t)
\right\},
\qquad w\in\{0,1\},
\end{aligned}
\end{equation}
where $S^w(x;t)$, $G^w(x;t)$, and $e^w(x)$ denote the nuisance functions corresponding to the conditional survival function $P(T_i>t\mid X_i=x,W_i=w)$, the conditional censoring survival function $P(\mathcal C_i>t\mid X_i=x,W_i=w)$, and the propensity score $P(W_i=w\mid X_i=x)$, respectively.  Here, for a fixed time point $t$,  the at-risk indicator
\(
Y_i(t)=\mathbb{I}(U_i>t).
\)
In practice, these nuisance functions are unknown and must be estimated from the data, and may be subject to model misspecification.

Despite this, the pseudo-outcome in \eqref{eq:pseudo} remains unbiased  under partial misspecification due to its double robustness. Building on this property, we estimate the pseudo-outcome by evaluating it over a prespecified grid of time points, reflecting that the target CATE is indexed by time and typically of interest over an interval rather than at a single time point. We then model its conditional expectation given baseline covariates using a deep neural network (DNN) introduced in the next subsection and obtain a flexible estimator of the time-varying CATE function.

\subsection{Deep Neural Network: Architecture and Implementation}
\label{DNN_architecture}

 Let $0 \le t_{\min} < t_{\max}$, where $[t_{\min}, t_{\max}]$ is a compact subset of the support of $U$. The time window $[t_{\min}, t_{\max}]$ is specified to balance clinical relevance and data support.  For example, in our motivating lung cancer study, \(t_{\min}=0\) represents the time of diagnosis, whereas \(t_{\max}\) is chosen to ensure that a sufficient number of patients remain at risk at that time for stable estimation. We estimate $\tau_0(x;t)$ on a prespecified grid of $J$ time points
\[
t_{\min}=t_1 < t_2 < \cdots < t_J=t_{\max}.
\]
This leads to the problem of learning a vector-valued function
\[
x \in \mathbb{R}^p \mapsto \tau_0(x) = \big(\tau_0(x;t_1),\ldots,\tau_0(x;t_J)\big)^\top \in \mathbb{R}^J,
\]
which we approximate using a deep neural network $f:\mathbb{R}^p \to \mathbb{R}^J$. {We consider a fully connected feedforward network with $L$ hidden layers. Let $p_0 = p$ denote the input dimension and $p_{L+1} = J$ denote the output dimension. For each hidden layer $\ell = 1,\ldots,L$, let $p_\ell$ denote its width.  Define the layer-wise transformations recursively. Set
\(
h^{(0)}(x) = x \in \mathbb{R}^{p_0}.
\)
For each hidden layer $\ell = 1,\ldots,L$, define
\[
h^{(\ell)}(x)
=
\sigma\!\big( W_\ell h^{(\ell-1)}(x) + b_\ell \big)
\in \mathbb{R}^{p_\ell},
\]
where
\(
W_\ell \in \mathbb{R}^{p_\ell \times p_{\ell-1}},
\,
b_\ell \in \mathbb{R}^{p_\ell},
\)
and $\sigma(u)=\max(u,0)$ is the RELU activation function, applied componentwise. The output layer is given by
\[
f(x)
=
W_{L+1} h^{(L)}(x) + b_{L+1}
\in \mathbb{R}^{p_{L+1}} = \mathbb{R}^J,
\]
where
\(
W_{L+1} \in \mathbb{R}^{J \times p_L},
\,
b_{L+1} \in \mathbb{R}^{J}.
\) Let $\mathcal{F}_n$ denote the class of such networks with the following structural constraints: 
\begin{itemize}
    \item[](\textbf{Depth}) $L \le L_n$;
\item[](\textbf{Width}) $\max_{1 \le \ell \le L} p_\ell \le p_n$;
  \item[](\textbf{Sparsity}) the total number of nonzero parameters satisfies
    \(
    \sum_{\ell=1}^{L+1} \|W_\ell\|_0 + \sum_{\ell=1}^{L+1} \|b_\ell\|_0 \le s_n;
    \)
  \item[](\textbf{Magnitude constraint}) all entries of $\{W_\ell, b_\ell\}$ are bounded in absolute value by $B_n$.
\end{itemize}}

{ With these structural constraints, regularization is implicitly achieved through sparsity and boundedness of the DNN class $\mathcal{F}_n$, which control model complexity and limit the magnitude of fitted values. These properties help stabilize estimation. } Another key feature of the DNN class $\mathcal{F}_n$ is its multi-output structure across time points, with hidden layers shared to capture temporal coherence in the CATE. Specifically, an $f(x) \in \mathcal{F}_n$ jointly estimates
\(
\big(\tau_0(x;t_1),\ldots,\tau_0(x;t_J)\big),
\)
thereby leveraging smoothness and dependence across time. The parameters $\{W_\ell, b_\ell\}_{\ell=1}^{L+1}$ are estimated by minimizing the empirical squared loss averaged over subjects and time points (Algorithm~\ref{alg:DSL}), with hyperparameters selected via cross-validation.

\subsection{Estimation and Implementation}
\label{est_and_implement}

Let
\(
\hat{\boldsymbol{\varphi}}_i
=
\bigl(\hat{\varphi}_i(t_1),\ldots,\hat{\varphi}_i(t_J)\bigr)^\top
\)
denote the vector of estimated pseudo-outcomes for subject \(i\)
on a prespecified grid of \(J\) time points,
\(
t_{\min}=t_1 < t_2 < \cdots < t_J=t_{\max}.
\)
 To reduce overfitting and permit flexible nuisance estimation, we employ
\(K\)-fold cross-fitting throughout. Let the observed data be
\(
\mathcal D=\{(U_i,\Delta_i,X_i,W_i)\}_{i=1}^N,
\)
 and  randomly partition \(\mathcal D\) into \(K\) approximately equal-sized, disjoint folds
\(
\mathcal D_1,\ldots,\mathcal D_K.
\) For each fold \(k\), let \(\mathcal D_{-k}=\mathcal D\setminus \mathcal D_k\) denote the corresponding training sample, and use \(\mathcal D_k\) as the held-out sample. 

{Within each training sample \(\mathcal D_{-k}\), we estimate the nuisance functions needed for pseudo-outcome construction. Specifically, for each treatment group \(w\in\{0,1\}\), we estimate 
\(
S^w(x;t),
\) and 
\(
G^w(x;t)
\)  using only subjects in \(\mathcal D_{-k}\) with \(W_i=w\). 
These nuisance components may be estimated using Cox-type models, or, when greater flexibility is desired, using nonparametric survival learners such as random survival forests. Also, the propensity score
\(
e^w(x) 
\)
is estimated on \(\mathcal D_{-k}\) using a classification method, such as logistic regression or a more flexible tree-based ensemble learner. } 

These estimators are then evaluated only on observations in the held-out fold \(\mathcal D_k\). Thus, for each subject \(i\in\mathcal D_k\), the pseudo-outcome vector
\(
\hat{\boldsymbol{\varphi}}_i
\)
is constructed using nuisance estimates obtained from data that do not include subject \(i\).  Repeating this procedure over all $K$ folds yields cross-fitted pseudo-outcomes for every subject in the sample.

We then estimate the CATE function $\tau(x;t_j)$, $j=1,\ldots,J$, by solving
\begin{equation}
\label{eq:main_estimator}
\hat\tau
=
\arg\min_{\tau \in \mathcal F_n}
\frac{1}{NJ}\sum_{i=1}^N
\bigl\|
\hat{\boldsymbol{\varphi}}_i - \tau(X_i)
\bigr\|_2^2.
\end{equation}
where $ \tau(X_i) = \bigl(\tau(X_i; t_1), \ldots, \tau(X_i; t_J) \bigr)^\top$. That is, we find the optimizer $\hat{\tau}(\cdot,\cdot)$ within the deep neural network class $\mathcal F_n$ with shared hidden layers and $J$ output nodes, each corresponding to a time point on the grid. The resulting procedure, referred to as the Deep Survival Learner (DSL), is summarized in Algorithm~\ref{alg:DSL}.

Regardless of the specific choice of nuisance estimators, our theoretical analysis shows that the error of the DNN estimator decomposes into nuisance estimation error as well as  approximation error and statistical error arising from fitting the DNN.

\begin{algorithm}[!htbp]
\caption{Deep Survival Learner (DSL)}
\label{alg:DSL}
\small
\begin{algorithmic}[1]
\Require Observed data $\mathcal D=\{(U_i,\Delta_i,X_i,W_i)\}_{i=1}^N$,
time grid $\{t_1,\ldots,t_J\}$, number of folds $K$
\Ensure Estimated CATE  $\hat\tau(\cdot;t)$ on $\{t_1,\ldots,t_J\}$

\State Randomly partition $\mathcal D$ into $K$ disjoint folds
$\mathcal D_1,\ldots,\mathcal D_K$

\For{$k=1,\ldots,K$}
    \State Let $\mathcal D_{-k}=\mathcal D\setminus \mathcal D_k$
    \State Estimate nuisance functions
    $\hat S^w(x;t)$, $\hat G^w(x;t)$, and $\hat e^w(x)$ using $\mathcal D_{-k}$
    \For{each $i\in\mathcal D_k$}
        \State Construct estimated pseudo-outcomes
        \(
        \hat{\boldsymbol{\varphi}}_i
        =
        \bigl(\hat{\varphi}_i(t_1),\ldots,\hat{\varphi}_i(t_J)\bigr)^\top
        \)
         using the fitted nuisance models
    \EndFor
\EndFor

\State Fit a multi-output DNN $\hat\tau \in \mathcal F_n$ by minimizing the objective in \eqref{eq:main_estimator}
\[
\frac{1}{NJ}\sum_{i=1}^N
\bigl\|
\hat{\boldsymbol{\varphi}}_i - \tau(X_i)
\bigr\|_2^2
\]

\end{algorithmic}
\end{algorithm}

\section{Theoretical Properties} \label{theory}
{\em Notation:} 
 For $w \in \{0,1\}$ and $t \in [t_{\min}, t_{\max}]$, define
\[
S_0^w(x;t) = P(T_i > t \mid X_i = x, W_i = w), 
e_0^w(x) = P(W_i = w \mid X_i = x), \]
\(G_0^w(x;t) = P(\mathcal{C}_i > t \mid X_i = x, W_i = w).
\)
Let $\hat S^w(x;t)$, $\hat e^w(x)$, and $\hat G^w(x;t)$ denote the corresponding estimators.  
For subject $i$, let
\(
O_i = (U_i, \Delta_i, X_i, W_i)
\)
denote the observed data. For any measurable function $g$, define the empirical measure by
\(
\mathbb{P}_n g = \frac{1}{n}\sum_{i=1}^n g(O_i).
\)
Further, for a random function $\hat g(\cdot)$ constructed independently of $O_i$, the quantities
\(
\E\{\hat g(O_i)\}
\quad\text{and}\quad
\E\{\hat g(O_i)\mid X_i\}
\)
are understood as expectations taken with respect to $O_i$ only, or
\(
\E\{\hat g(O_i)\}
=
\int \hat g(o)\, dP_{O}(o), \) and 
\(
\E\{\hat g(O_i)|X_i=x\}
=
\int \hat g(o)\, dP_{O\mid X}(o|x).
\)
For a vector \(x=(x_1,\ldots,x_d)^\top\in\mathbb R^d\), define its \(\ell_q\)-norm by
\(
\|x\|_q
=
\left(\sum_{j=1}^d |x_j|^q\right)^{1/q},
\, 1\le q<\infty,
\)
and
\(
\|x\|_\infty
=
\max_{1\le j\le d}|x_j|.
\)
 Let \(P_X\) denote the probability distribution of the covariate vector \(X\) on \(\mathcal X\). For a measurable function \(f:\mathcal X\to\mathbb R\), define
\(
\|f\|_{L_q(P_X)}
=
\left(\int_{\mathcal X} |f(x)|^q\, dP_X(x)\right)^{1/q},
\, 1\le q<\infty,
\)
and
\(
\|f\|_{L_\infty(P_X)}
=
\sup_{x\in\mathcal X} |f(x)|.
\)

 We impose the following assumptions to establish the statistical properties of the proposed estimator.

\begin{assumption}[Consistency]\label{assum_consistency}
For each individual,
\(
T_i = T_i^{W_i}.
\)
\end{assumption}

\begin{assumption}[Unconfoundedness]\label{assum_unconfoundedness}
\(
\{T_i^0, T_i^1\} \indep W_i \mid X_i.
\)
\end{assumption}

\begin{assumption}[Overlap]\label{assum_overlap}
 Assume the support of $X$ is $\mathcal X=[0,1]^d$, and $P_X$ has a density bounded away from $0$ and $\infty$ on $\mathcal X$.
There exists a constant $c > 0$ such that \(
e_0^w(x) \ge c
\), for all $x \in \mathcal X$ and $w \in \{0,1\}$.
\end{assumption}

\begin{assumption}[Noninformative censoring]\label{assum_indep_censoring}
For each $w \in \{0,1\}$, we assume that
\(
\mathcal C_i \;\indep\; T_i^w \mid (X_i, W_i = w).
\)
\end{assumption}

\begin{assumption}[Censoring positivity]\label{assum_positivity}
There exists a constant $c > 0$ such that \(
G_0^w(x; t) \ge c
\), for all $(x,t)$ in the support of $(X_i,T_i)$ and all $w \in \{0,1\}$.
 
\end{assumption}

\begin{assumption}[Uniform convergence rate]\label{assum_uni_consistency} For a fixed $t \in [t_{\min}, t_{\max}]$, assume the estimates, corresponding to $S_0^w(x;t), G_0^w(x;t) $ and $e_0^w(x)$,
satisfy
    \[
    \sup_x |\hat S^w(x;t)-S_0^w(x;t)|=O_p(r_n^{(S)}), \qquad
    \sup_x |\hat G^w(x;t)-G_0^w(x;t)|=O_p(r_n^{(G)}),
    \]
    \[
    \sup_x |\hat e^w(x)-e_0^w(x)|=O_p(r_n^{(e)}),
    \]
    for some positive sequences $r_n^{(S)},
r_n^{(G)}, r_n^{(e)}$ that converge to 0.
    
\end{assumption}

Assumptions~\ref{assum_consistency}--\ref{assum_overlap} are standard in causal inference and ensure identification of the potential outcome distributions \citep{rosenbaum1983propensity, imbens2015causal}. Assumptions~\ref{assum_indep_censoring}--\ref{assum_positivity} are commonly imposed in survival analysis to guarantee identifiability of the survival and censoring mechanisms \citep{andersen1982cox, tsiatis2006semiparametric}. Assumption~\ref{assum_uni_consistency} requires uniform convergence of the estimated nuisance functions, which is standard in semiparametric and machine learning-based inference and determines the convergence rate of the proposed estimator \citep{vanderVaart1998, chernozhukov2018dml}.

\begin{theorem}[Identification]\label{thm:ID}
Under Assumptions~\ref{assum_consistency}-\ref{assum_positivity}, if all nuisance functions in \eqref{eq:pseudo} equal their true counterparts, it holds that
\[
\E\{\varphi_i(t)\mid X_i=x\}=\tau_0(x;t).
\]
\end{theorem}

 Theorem~\ref{thm:ID}  establishes identification of $\tau_0(x;t)$ via the conditional expectation of $\varphi_i(t)$ given $X_i$, thereby reducing the problem of estimating heterogeneous treatment effects under censoring to a conditional mean regression problem.
 
\begin{theorem}[Double robustness]\label{thm:DR}
Under Assumptions~\ref{assum_consistency}-\ref{assum_positivity}, suppose the censoring model is correctly specified, i.e.,
\(
G^w(x;t)=P(\mathcal C_i>t\mid X_i=x,W_i=w).
\)
If either
 the survival model is correctly specified, i.e.,
\(
S^w(x;t)=P(T_i>t\mid X_i=x,W_i=w),
\)
or
 the propensity score model is correctly specified, i.e.,
\(
e^w(x)=P(W_i=w\mid X_i=x),
\)
it holds that
\[
\mathbb E\{\varphi_i(t)\mid X_i=x\}=\tau_0(x;t).
\]
\end{theorem}

 The double robustness of the proposed pseudo-outcome, established in Theorem~\ref{thm:DR}, is especially appealing in observational survival studies, where it is often difficult to correctly specify all nuisance components.

 The consistency of the final estimator further relies on the approximation and estimation properties of the DNN.
We establish consistency of $\hat{\tau}(x;t)$ for $\tau_0(x;t)$ at a fixed time point $t$ (i.e., $J=1$); the extension to $J>1$ is analogous. Recall that
\(
\tau_0(x;t)
=
\mathbb{E}\big[\varphi_i^0(t)\mid X_i = x\big],
\)
where $\varphi_i^0(t)$ is the oracle pseudo-outcome and $\hat{\varphi}_i(t)$ is its cross-fitted counterpart (Section~\ref{est_and_implement}). The DNN estimator is obtained via
\[
\hat\tau(\cdot;t) = \arg\min_{f\in\mathcal F_n}
\mathbb P_n\big[\{\hat\varphi_i(t)-f(X_i)\}^2\big],
\]
where $\mathcal F_n$ is the DNN class. 
We begin with the decomposition
\[
\hat{\tau}(x;t) - \tau_0(x;t)
=
\underbrace{\hat{\tau}(x;t) - \mathbb{E}\big[\hat{\varphi}_i(t)\mid X_i = x\big]}_{\text{regression error}}
+
\underbrace{\mathbb{E}\big[\hat{\varphi}_i(t) - \varphi_i^0(t)\mid X_i = x\big]}_{\text{nuisance-induced error}}.
\]
The first term reflects the error from nonparametric regression, while the second term captures the impact of estimating nuisance functions. We bound the second term using the following results. 

\begin{theorem}[Bias decomposition with estimated nuisance functions]\label{thm:bias_full}
     Suppose Assumptions~\ref{assum_overlap},~\ref{assum_positivity}, and~\ref{assum_uni_consistency} hold. Let $\hat{\varphi}_i(t)$ and $\varphi_i^0(t)$ denote the estimated and population pseudo-outcomes, respectively, obtained by plugging the cross-fitted and true nuisance functions into \eqref{eq:pseudo}.
        Then,  it holds that
    \begin{align}
        \mathbb E\big[\hat\varphi_i(t)-\varphi_i^0(t)\mid X_i=x\big]
        &=
        \sum_{w=0}^1 (-1)^{w+1}
        \Bigg[
        \frac{\hat e^w(x)-e_0^w(x)}{\hat e^w(x)}
        \big\{\hat S^w(x;t)-S_0^w(x;t)\big\}
        \notag\\
        &
        +
        \frac{e_0^w(x)}{\hat e^w(x)}
        \frac{S_0^w(x;t)}{\hat G^w(x;t)}
        \big\{G_0^w(x;t)-\hat G^w(x;t)\big\}
        \Bigg].
        \label{eq:full_bias_decomp}
    \end{align}
    Consequently, it follows that
    \[
    \sup_x
    \left|
    \mathbb E\big[\hat\varphi_i(t)-\varphi_i^0(t)\mid X_i=x\big]
    \right|
    =
    O_p\!\left(r_n^{(e)}r_n^{(S)} + r_n^{(G)}\right).
    \]
    
    In particular, if \(\hat G^w(x;t)=G_0^w(x;t)\), then we have 
    \[
    \sup_x
    \left|
    \mathbb E\big[\hat\varphi_i(t)-\varphi_i^0(t)\mid X_i=x\big]
    \right|
    =
    O_p\!\left(r_n^{(e)}r_n^{(S)}\right).
    \]
    
\end{theorem}

 We impose smoothness conditions on the CATE function in \(x\) and additional regularity conditions on the cross-fitted pseudo-outcomes and the network class, which are required for the consistency and error rate results at a fixed time point.
 \begin{definition}[H\"older class]\label{def:holder}
Let \(\mathcal X=[0,1]^d\) and let \(\beta>0\). Write
\[
\beta=q+\alpha,
\qquad
q=\lfloor \beta\rfloor,
\qquad
\alpha\in(0,1],
\]
where $q=\lfloor \beta\rfloor$ is the integer part of $\beta$. Let \(\nu=(\nu_1,\ldots,\nu_d)\),  where each \(\nu_j\) is a nonnegative integer. 
For \(x=(x_1,\ldots,x_d)\), define
\(
D_x^\nu
=
\frac{\partial^{\|\nu\|_1}}{\partial x_1^{\nu_1}\cdots \partial x_d^{\nu_d}}.
\)
For \(M>0\), define \(\mathcal H^\beta(M)\) as the class of functions
\(
f:\mathcal X\to\mathbb R
\)
such that, for all mixed partial derivatives \(D_x^\nu f(x)\) with \(\|\nu\|_1\le q\), these derivatives exist and are uniformly bounded by \(M\), that is,
\[
\max_{\|\nu\|_1\le q}\sup_{x\in\mathcal X}|D_x^\nu f(x)|\le M,
\]
and moreover,
\[
\sup_{x\neq y}
\frac{|D_x^\nu f(x)-D_x^\nu f(y)|}{\|x-y\|_2^\alpha}
\le M,
\qquad \|\nu\|_1=q.
\]
\end{definition}
\begin{assumption}[Bounded second moments]\label{assum_moment} The cross-fitted 
$\hat\varphi_i$ satisfies
 \(
        \sup_{ 1\le i \le n} \mathbb E\big[\hat\varphi_i^2(t)\big] < \infty.
        \)
\end{assumption}

\begin{assumption}[Network approximation capacity] \label{assum_capacity}
The network class $\mathcal F_n$ is rich such that there exists $f_n^\star\in\mathcal F_n$ satisfying
    \[
    \|f_n^\star-\tau_0(\cdot;t)\|_{L_\infty(P_X)}
    \le C n^{-\beta/(2\beta+d)}.
    \]
 \end{assumption}

\begin{assumption}[Uniform convergence over network] \label{assum_emp}
 The empirical squared-loss process satisfies
    \[
    \sup_{f\in\mathcal F_n}
    \left|
    \mathbb P_n\big[\{\hat\varphi_i(t)-f(X_i)\}^2\big]
    -
    \mathbb E\big[\{\hat\varphi_i(t)-f(X_i)\}^2\big]
    \right|
    =
    O_p(\rho_n),
    \]
    for some $\rho_n\to 0$.  
\end{assumption}

\begin{assumption}[Controlled optimization error]
\label{assum_opt}  The optimization error is controlled in the sense that
    \[
    \mathbb P_n\big[\{\hat\varphi_i(t)-\hat\tau(X_i;t)\}^2\big]
    \le
    \inf_{f\in\mathcal F_n}
    \mathbb P_n\big[\{\hat\varphi_i(t)-f(X_i)\}^2\big]
    +\delta_n^2,
    \]
    where $\delta_n=o_p(1)$.
\end{assumption}

 The H\"older class in Definition~\ref{def:holder} is a standard smoothness condition in nonparametric regression used to characterize approximation error and minimax rates \citep{tsybakov2009introduction,gyorfi2006distribution}. Assumption~\ref{assum_moment} imposes a bounded second moment condition on the pseudo-outcomes, ensuring well-behaved risk and concentration properties \citep{chernozhukov2018double}. Assumption~\ref{assum_capacity} requires that the network class $\mathcal F_n$ can approximate $\tau_0(\cdot;t)$ at the optimal rate, consistent with approximation results for ReLU networks under H\"older smoothness \citep{schmidt2020deep}. Assumption~\ref{assum_emp} ensures uniform convergence of the empirical loss to its population counterpart over $\mathcal F_n$, providing control of the estimation error in empirical risk minimization \citep{van1996weak,bartlett2017spectrally}. Finally, Assumption~\ref{assum_opt} controls the optimization error arising from approximate minimization of the empirical loss \citep{schmidt2020deep}.

\begin{theorem}[Consistency of the constructed estimator at a fixed $t$]\label{thm:consistency_tailored}
    Suppose  Assumptions~\ref{assum_consistency}-- \ref{assum_emp} hold and   that the convergence rates of nuisance functions in Assumption \ref{assum_uni_consistency}  satisfy that 
    \(
    r_n^{(e)}r_n^{(S)}+r_n^{(G)}\to 0.
    \) 
         It follows that
    \[
    \|\hat\tau(\cdot;t)-\tau_0(\cdot;t)\|_{L_2(P_X)}\xrightarrow{p}0.
    \].
\end{theorem}

\begin{theorem}[Rate for the constructed CATE estimator at a fixed $t$]\label{thm:overall_rate_holder}
Suppose 
    \(
\tau_0(\cdot;t)\in\mathcal H^\beta(M)
    \)
    for some $\beta>0$ and $M>0$, and that 
 Assumptions~\ref{assum_consistency}--\ref{assum_opt} hold.  
It follows that 
\[
\|\hat\tau(\cdot;t)-\tau_0(\cdot;t)\|_{L_2(P_X)}
=
O_p\!\left(
n^{-\beta/(2\beta+d)}
+
\rho_n^{1/2}
+
\delta_n
+
r_n^{(e)}r_n^{(S)}
+
r_n^{(G)}
\right).
\]
In particular, if $\rho_n\lesssim n^{-2\beta/(2\beta+d)}$,  it holds that
\[
\|\hat\tau(\cdot;t)-\tau_0(\cdot;t)\|_{L_2(P_X)}
=
O_p\!\left(
n^{-\beta/(2\beta+d)}
+
\delta_n
+
r_n^{(e)}r_n^{(S)}
+
r_n^{(G)}
\right).
\]
\end{theorem}

\begin{remark} Cross-fitting ensures that the pseudo-outcome $\hat{\varphi}_i(t)$ is constructed using nuisance estimators that are independent of observation $i$. This allows us to treat the estimated nuisance functions as fixed when taking conditional expectations, which is critical for establishing the bias decomposition in Theorem~\ref{thm:bias_full}.
In particular, it justifies the expansion
\(
\mathbb{E}\{\hat{\varphi}_i(t)\mid X_i=x\},
\)
leading to the second-order term $r_n^{(e)} r_n^{(S)}$ arising from the product of estimation errors in the propensity score and outcome model. Without cross-fitting, additional first-order bias terms of order $r_n^{(e)} + r_n^{(S)}$ may appear due to overfitting.
We note that estimation of the censoring model introduces an additional first-order term $r_n^{(G)}$, as reflected in the bias decomposition.
\end{remark}

 To quantify the gain from joint estimation of \(\{\tau_0(x;t_j)\}_{j=1}^J\), we impose additional regularity on the CATE surface in the time dimension. In contrast to the single-time-point setting, where only smoothness in \(x\) is required, joint estimation exploits cross-time dependence by modeling \(\tau_0(x;t)\) as a function that is jointly smooth in \((x,t)\).
We formalize this through a bivariate H\"older class and prove  that it induces a low-complexity representation in the time direction and enables control of the joint approximation error. The following lemma establishes this representation, and the subsequent theorem uses it to characterize the advantage of shared network estimation across time. Without loss of generality, we take
$[t_{\min}, t_{\max}]=[0,1]$ throughout the  development.
\begin{definition}[H\"older smoothness in \((x,t)\)]\label{def:holder_xt}
Let \(\mathcal X=[0,1]^d\) and \(\mathcal T=[0,1]\). For \(\beta_x>0\) and \(\beta_t>0\), write
\[
\beta_x = q_x+\alpha_x,
\qquad
q_x=\lfloor \beta_x\rfloor,
\qquad
\alpha_x\in(0,1],
\]
and
\[
\beta_t = q_t+\alpha_t,
\qquad
q_t=\lfloor \beta_t\rfloor,
\qquad
\alpha_t\in(0,1].
\]
Let \(\nu=(\nu_1,\ldots,\nu_d)\),  where  \(\nu_j\) is a nonnegative integer. For \(x=(x_1,\ldots,x_d)\), define
\(
D_x^\nu
=
\frac{\partial^{\|\nu\|_1}}{\partial x_1^{\nu_1}\cdots \partial x_d^{\nu_d}}.
\)
For \(M>0\), define \(\mathcal H^{\beta_x,\beta_t}(M)\) as the class of functions
\(
f:\mathcal X\times\mathcal T \to \mathbb R
\)
such that, for all mixed derivatives
\[
D_x^\nu \partial_t^m f(x,t),
\qquad \|\nu\|_1\le q_x,\quad  m=0,1,\ldots,q_t,
\]
these derivatives exist and are uniformly bounded by \(M\), and moreover satisfy
\[
\sup_{t\in\mathcal T}\sup_{x\neq y}
\frac{\big|D_x^\nu \partial_t^m f(x,t)-D_x^\nu \partial_t^m f(y,t)\big|}{\|x-y\|_2^{\alpha_x}}
\le M,
\qquad \|\nu\|_1=q_x,\quad m=0,1,\ldots,q_t, 
\]
and
\[
\sup_{x\in\mathcal X}\sup_{s\neq t}
\frac{\big|D_x^\nu \partial_t^{q_t} f(x,t)-D_x^\nu \partial_t^{q_t} f(x,s)\big|}{|t-s|^{\alpha_t}}
\le M,
\qquad \|\nu\|_1\le q_x.
\]
\end{definition}


 \begin{lemma}[Basis expansion under H\"older smoothness]\label{lem:time_basis_holder}
Suppose
\(
\tau_0(\cdot; \cdot) \in \mathcal H^{\beta_x,\beta_t}(M)
\)
in the sense of Definition~\ref{def:holder_xt}. Then for every integer \(K_n\ge 1\), there exist basis functions
\(\psi_1,\ldots,\psi_{K_n}\) on \([0,1]\) and coefficient functions
\(b_1,\ldots,b_{K_n} \in \mathcal H^{\beta_x}(M):\mathcal X\to\mathbb R\) such that
\[
\sup_{1\le j\le J}
\left\|
\tau_0(\cdot;t_j)-\sum_{m=1}^{K_n} b_m(\cdot)\psi_m(t_j)
\right\|_{L_2(P_X)}^2
\le
C K_n^{-2\beta_t},
\]
where \(C>0\) depends only on \(M\) and \(\beta_t\).
Consequently, it holds that
\[
\sum_{j=1}^J
\left\|
\tau_0(\cdot;t_j)-\sum_{m=1}^{K_n} b_m(\cdot)\psi_m(t_j)
\right\|_{L_2(P_X)}^2
\le
C J K_n^{-2\beta_t}.
\]
\end{lemma}

  \begin{theorem}[Error rate for joint estimation over time]\label{thm:joint_pseudo}
 Let the joint estimator
\(
\hat{\boldsymbol\tau}(x)
=
\bigl(\hat\tau(x;t_1),\ldots,\hat\tau(x;t_J)\bigr)^\top, {\rm with} \,  J>1, 
\)
be defined by
\[
\hat{\boldsymbol\tau}
=
\arg\min_{f\in\mathcal F_n}
\mathbb P_n\!\left[
\sum_{j=1}^J \{\hat\varphi_i(t_j)-f_j(X_i)\}^2
\right],
\]
where \(\mathcal F_n\) is a class of \(J\)-output neural networks with shared hidden layers. Suppose \(
    \tau_0 \in \mathcal H^{\beta_x,\beta_t}(M).
    \)
Under Assumptions~\ref{assum_consistency}--\ref{assum_opt},   
 it holds that 
\begin{align*}
& \sum_{j=1}^J
\|\hat\tau(\cdot;t_j)-\tau_0(\cdot;t_j)\|_{L_2(P_X)}^2
= \\
& O_p\!\left(
J^{1/(2\beta_t+1)}
n^{-\,4\beta_t\beta_x/\{(2\beta_t+1)(2\beta_x+d)\}}
+
\left(
J\,n^{2\beta_x/(2\beta_x+d)}
\right)^{1/(2\beta_t+1)}\rho_n
+
\delta_{n,J}^2
+
J\{r_n^{(e)}r_n^{(S)}+r_n^{(G)}\}^2
\right),
\end{align*}
 where $\delta_{n,J}$ is the optimization error defined in the proof.

\end{theorem}
\begin{remark} 
  If one instead fits \(J\) separate DNNs, one for each time point \(t_j\), without using smoothness in \(t\), then the leading order of  approximation error is 
\(
n^{-2\beta_x/(2\beta_x+d)}
\)
for each \(j\) by Theorem \ref{thm:overall_rate_holder}, so the aggregate approximation error is of order
\(
J\, n^{-2\beta_x/(2\beta_x+d)}.
\)
By contrast, the joint estimator achieves a leading order of approximation error
\(
J^{1/(2\beta_t+1)}
n^{-\,4\beta_t\beta_x/\{(2\beta_t+1)(2\beta_x+d)\}}.
\)
Since \(1/(2\beta_t+1)<1\) for \(\beta_t>0\), the dependence on \(J\) is strictly sublinear. Therefore, the joint estimator achieves a smaller approximation rate than fitting \(J\) separate DNNs independently.

\end{remark}

\section{Simulation Study}
\label{simulation}
To evaluate the finite-sample performance of DSL and to empirically verify its double robustness property, we conducted simulation studies under settings where either the survival model or the propensity score model is misspecified. We compared DSL with several commonly used causal machine learning methods, including the M-learner~\citep{bo2025evaluating}, X-learner~\citep{kunzel2019metalearners}, and R-learner~\citep{nie2021quasi}. {We also included two doubly robust learners based on nonparametric failure time Bayesian additive regression trees (NFT-BART), namely DRL-NFT-BART and PSDRL-NFT-BART~\citep{yang2025doubly}, and  the causal survival forest~\citep{cui2023estimating}.}

Across all scenarios, we generated independent datasets consisting of 800 training samples and 400 testing samples. For each subject $i$, a 30-dimensional covariate vector $X_i$ was generated independently from a uniform distribution on $[-3,3]^{30}$.

\textit{Case 1: Correct survival and propensity score models} In this setting, both the survival model and the propensity score model were correctly specified. The treatment assignment $W_i$ was generated from a logistic regression model,
\[
\operatorname{logit}\{P(W_i = 1 \mid X_i )\} = \alpha^{\top}X_i,
\]
where $\operatorname{logit}(p) = \log\{p/(1-p)\}$. The survival time followed a Weibull-Cox proportional hazards model,
\[
\lambda_i(t) = \lambda \nu t^{\nu-1}
\exp\big(\beta^{\top}X_i + W_i\,\theta^{\top}X_i\big),
\]
and the censoring time followed an Exponential-Cox model,
\[
\lambda_i^C(t) = \lambda^C \exp(\gamma^{\top}X_i).
\]
We set $\lambda = 0.1$, $\nu = 1.2$, and $\lambda^C = 0.2$, so the average censoring rate is approximately 35\%. The coefficient vectors
$\alpha, \beta, \theta, \gamma \in \mathbb{R}^{30}$
were generated as evenly spaced values on the intervals
$[-1,1]$, $[0.4,-0.2]$, $[0.3,-0.1]$, and $[0.3,-0.3]$, respectively. {Additional simulation results under a simplified setting with lower-dimensional covariates and a larger sample size are provided in the appendix for comparison.}

\textit{Case 2: Correct survival model, misspecified propensity score} This setting evaluates robustness to misspecification of the propensity score model. The survival and censoring mechanisms were identical to those in Case 1, while the treatment assignment followed a nonlinear model,
\[
\operatorname{logit}\{P(W_i = 1 \mid X_i)\}
= 0.5
+ 0.45\tanh(\eta_i/1.2)
+ 0.25\mathbb{I}(\eta_i > 0)
- 0.15\mathbb{I}(\eta_i < 1)
+ 0.1\sin(\eta_i^2/3),
\]
where $\eta_i = \alpha^{\top}X_i$ and $\alpha$ is defined as in Case 1. The resulting censoring rate was approximately 40\%.

\textit{Case 3: Misspecified survival model, correct propensity score} This setting evaluates robustness to misspecification of the survival model. The propensity score was generated as in Case 1, while the survival time followed a non-proportional hazards model with time-varying treatment effects,
\[
\lambda_i(t) = \lambda_0
\exp\left\{
g(X_i)
+ W_i \cdot h(X_i)\cdot \sin\!\big(v(X_i)\, t\big)
\right\},
\]
where
\[
g(X_i) = \beta^{\top}X_i,\quad
u(X_i) = \gamma^{\top}X_i,\quad
v(X_i) = \theta^{\top}X_i,
\]
\[
h(X_i)
= 0.5\cos\big(u(X_i)\big)
+ 0.2u(X_i)^2.
\]

The coefficient vectors $\beta, \gamma, \theta$ were generated as evenly spaced values on the intervals
$[-0.3,0.5]$, $[-0.2,0.3]$, and $[-0.2,0.4]$, respectively. We set $\lambda_0 = 0.05$.
The censoring mechanism was the same as in Case~1, except that $\lambda^C = 0.05$, yielding a censoring rate of approximately 50\%.

In each simulation setting, we evaluated CATE over a grid of $J$ time points, with
\(
J \in \{1, 50, 100\}.
\)
The time grid consisted of $J$ equally spaced points between the 20th and 80th percentiles of the observed survival time distribution. When $J=1$, the evaluation time corresponded to the median observed survival time.

We generated $200$ independent replicates. For each replicate, DSL was trained on the training sample and evaluated on the testing sample. For replicate $r=1, \ldots, 200$, we computed the bias and mean squared error (MSE) averaged over test subjects and time points,
\[
\mathrm{Bias}^{(r)}
=
\frac{1}{NJ}\sum_{i=1}^n\sum_{j=1}^J
\Big\{
\hat\tau^{(r)}(X_i^{(r)};t_j^{(r)})
-
\tau_0(X_i^{(r)};t_j^{(r)})
\Big\},
\]
\[
\mathrm{MSE}^{(r)}
=
\frac{1}{NJ}\sum_{i=1}^n\sum_{j=1}^J
\Big\{
\hat\tau^{(r)}(X_i^{(r)};t_j^{(r)})
-
\tau_0(X_i^{(r)};t_j^{(r)})
\Big\}^2,
\]
where $N = 400$ denotes the number of test samples. We report the average bias and MSE across replicates,
 along with their corresponding 95\% Monte Carlo confidence intervals, in
Tables~\ref{simulate_performance_bias} and~\ref{simulate_performance_mse}.

We employed five-fold cross-fitting to construct the estimated pseudo-outcomes. The nuisance functions for the survival and censoring time was estimated using Cox proportional hazards models with linear predictors, and the treatment assignment was modeled using a logistic regression with linear covariates. {The DSL architecture consists of three fully connected hidden layers with 128, 64, and 32 neurons, respectively, each using ReLU activation functions}, and was trained using the Adam optimizer with a learning rate of 0.003. Regularization is achieved through the architectural design of the network, including its depth and width, which implicitly control model complexity. All neural network models were implemented in \texttt{PyTorch}. For competing methods, the M-Learner used logistic regression (via \texttt{glm}) for propensity score estimation and generalized boosted regression models (GBM; \texttt{gbm} package) for outcome regression, with inverse probability of censoring weights estimated using random survival forests. The X-Learner employed survival random forests implemented in \texttt{randomForestSRC}, with tuning over terminal node sizes from 1 to 100 and candidate splitting variables starting at 15. The outcome regressions for the X-Learner were also fitted using GBM. {The Causal Survival Forest was implemented using the \texttt{grf} package with 2000 trees and a minimum node size of 5.   Finally, we implemented two doubly robust learners, DRL-NFT-BART and SDRL-NFT-BART\citep{yang2025doubly} via \texttt{nftbart}.  DRL-NFT-BART combines inverse probability of censoring weighting with outcome regression, while PSDRL-NFT-BART uses a propensity-score-weighted pseudo-outcome. Propensity scores were estimated via Super Learner, censoring probabilities via random survival forests, and NFT-BART was fitted with 1000 burn-in and 2000 posterior samples; CATE was estimated using a regression forest with 2000 trees.}

 {Table~\ref{simulate_performance_bias} summarizes the mean bias and corresponding 95\% Monte Carlo confidence intervals across all simulation settings. Overall, the proposed DSL consistently exhibits the smallest bias among competing approaches. When both the survival and treatment models are correctly specified, the confidence intervals for DSL include zero at \(J=1\) and \(J=50\), indicating negligible bias; similar behavior persists when only the survival model is correctly specified, with coverage maintained across all \(J\). In contrast, competing learners, including M-, X-, and R-learners, exhibit substantial and persistent bias that does not diminish with increasing \(J\). Among the doubly robust methods, DRL-NFT-BART shows partial robustness, with bias substantially reduced under treatment model misspecification but remaining non-negligible when the survival model is misspecified. PSDRL-NFT-BART, which relies on propensity-score-weighted pseudo-outcomes, exhibits consistently large bias across all scenarios, indicating limited robustness in finite samples. When only the treatment model is correctly specified, DSL continues to outperform all competitors, achieving markedly smaller bias. For example, at \(J=100\), DSL attains a bias of 0.011 (95\% CI: 0.005, 0.016), compared with 0.219 (95\% CI: 0.203, 0.234) for the R-learner and 0.236 (95\% CI: 0.217, 0.255) for DRL-NFT-BART.}

    {Table~\ref{simulate_performance_mse} reports the corresponding mean squared error (MSE). Across all scenarios, DSL achieves competitive performance, with only a mild increase in MSE as the number of time points \(J\) grows. In the fully correctly specified setting, DRL-NFT-BART and the X-learner attain slightly lower MSE than DSL, although the differences are modest. Under nuisance misspecification, DSL demonstrates a clear advantage. When only the survival model is correctly specified, DSL substantially outperforms all competing methods, with MSE increasing only slightly from 0.070 (95\% CI: 0.068, 0.072) at \(J=1\) to 0.076 (95\% CI: 0.074, 0.077) at \(J=100\), compared to higher MSEs for alternatives such as the X-learner. When the survival model is misspecified, DRL-NFT-BART achieves lower MSE at small \(J\), but this advantage diminishes as \(J\) increases, and DSL remains competitive. In contrast, PSDRL-NFT-BART exhibits consistently higher MSE across all settings, consistent with its large bias.}

{ In summary, DSL remains competitive under correct specification, incurring minimal efficiency loss, while offering clear advantages under nuisance misspecification through reduced bias and enhanced robustness, particularly as model complexity increases.}

\begin{table}[ht]
    \caption{Bias of conditional average treatment effect estimators across time under nuisance model misspecification}
    \label{simulate_performance_bias}
    \begin{threeparttable}
    \centering
    \makebox[\textwidth]{\begin{tabular}{llccc}
        \hline
    &  & \shortstack{Case 1: Correct Survival\\ Correct Treatment} & \shortstack{Case 2: Correct Survival\\ Wrong Treatment} & \shortstack{Case 3: Wrong Survival\\ Correct Treatment} \\
       \hline
    \multirow{5}{*}{\(J = 1\)}
    &Deep Survival Learner & 0.002(-0.003, 0.007) & 0.003(-0.003, 0.010) & 0.012(0.005, 0.018)\\
    &M-Learner             & 0.199(0.193, 0.204)  & 0.221(0.213, 0.229) & 0.305(0.300, 0.310)\\
    &R-Learner             & 0.068(0.054, 0.082)  & 0.030(0.011, 0.049) & 0.223(0.207, 0.239)\\
    &X-Learner             & 0.204(0.198, 0.211)  & 0.221(0.210, 0.232) & 0.297(0.289, 0.305)\\
    &DRL-NFT-BART          & 0.177(0.161, 0.194)  & -0.033(-0.067, 0.001) & 0.209(0.188, 0.231)\\
    &PSDRL-NFT-BART        & 0.360(0.354, 0.367)  & 0.372(0.364, 0.380) & 0.386(0.380, 0.391)\\
    &Causal Survival Forest& 0.378(0.371, 0.384)  & 0.392(0.382, 0.402) & 0.475(0.468, 0.481)\\
           \hline
    \multirow{5}{*}{\(J = 50\)}
    &Deep Survival Learner & -0.001(-0.006, 0.004) & 0.001(-0.006, 0.007) & 0.012(0.006, 0.018)\\
    &M-Learner             & 0.199(0.194, 0.204) & 0.221(0.213, 0.229) & 0.305(0.299, 0.310)\\
    &R-Learner             & 0.068(0.054, 0.081) & 0.022(0.003, 0.041) & 0.220(0.205, 0.234)\\
    &X-Learner             & 0.206(0.200, 0.213) & 0.226(0.216, 0.236) & 0.300(0.292, 0.308)\\
    &DRL-NFT-BART          & 0.167(0.147, 0.187) & -0.024(-0.056, 0.009) & 0.238(0.221, 0.255)\\
    &PSDRL-NFT-BART        & 0.360(0.354, 0.366) & 0.373(0.364, 0.381) & 0.386(0.380, 0.392)\\
    &Causal Survival Forest& 0.378(0.372, 0.384) & 0.393(0.383, 0.403) & 0.474(0.468, 0.481)\\
             \hline
    \multirow{5}{*}{\(J = 100\)}
    &Deep Survival Learner & -0.007(-0.012, -0.002) & -0.001(-0.008, 0.005) & 0.011(0.005, 0.016)\\
    &M-Learner             & 0.198(0.193, 0.204) & 0.221(0.213, 0.229) & 0.305(0.299, 0.310)\\
    &R-Learner             & 0.073(0.061, 0.086) & 0.028(0.010, 0.047) & 0.219(0.203, 0.234)\\
    &X-Learner             & 0.205(0.198, 0.212) & 0.229(0.219, 0.240) & 0.295(0.286, 0.303)\\
    &DRL-NFT-BART          & 0.158(0.136, 0.180) & -0.037(-0.071, -0.002) & 0.236(0.217, 0.255)\\
    &PSDRL-NFT-BART        & 0.360(0.354, 0.367) & 0.373(0.364, 0.381) & 0.384(0.378, 0.389)\\
    &Causal Survival Forest& 0.378(0.372, 0.384) & 0.393(0.384, 0.403) & 0.474(0.468, 0.481)\\
             \hline
        \end{tabular}}
     \begin{tablenotes}\footnotesize
\item \textsuperscript{1} Mean bias and 95\% Monte Carlo confidence intervals are computed based on 200 simulation replicates.
\item \textsuperscript{2} Conditional average treatment effects are evaluated at \(J \in \{1, 50, 100\}\) time points, equally spaced between the 20th and 80th percentiles of the observed survival time distribution.
\item \textsuperscript{3} Columns correspond to combinations of correctly and incorrectly specified survival and treatment assignment models; the censoring model is correctly specified in all scenarios.
\end{tablenotes}
    \end{threeparttable}
\end{table}

\begin{table}[ht]
    \caption{Mean squared error of conditional average treatment estimators across time under nuisance model misspecification}
    \label{simulate_performance_mse}
    \begin{threeparttable}
    \centering
    \makebox[\textwidth]{\begin{tabular}{llccc}
        \hline
    &  & \shortstack{Case 1: Correct Survival\\ Correct Treatment} & \shortstack{Case 2: Correct Survival\\ Wrong Treatment} & \shortstack{Case 3: Wrong Survival\\ Correct Treatment} \\
       \hline
    \multirow{5}{*}{\(J = 1\)}
    &Deep Survival Learner & 0.077(0.075, 0.080) & 0.070(0.068, 0.072) & 0.104(0.101, 0.106)\\
    &M-Learner             & 0.114(0.110, 0.117) & 0.122(0.117, 0.127) & 0.190(0.185, 0.194)\\
    &R-Learner             & 0.174(0.168, 0.180) & 0.235(0.227, 0.244) & 0.260(0.251, 0.269)\\
    &X-Learner             & 0.081(0.078, 0.084) & 0.093(0.089, 0.098) & 0.154(0.150, 0.159)\\
    &DRL-NFT-BART          & 0.070(0.064, 0.075) & 0.133(0.113, 0.153) & 0.087(0.077, 0.097)\\
    &PSDRL-NFT-BART        & 0.142(0.138, 0.146) & 0.152(0.147, 0.157) & 0.152(0.148, 0.156)\\
    &Causal Survival Forest& 0.156(0.151, 0.160) & 0.170(0.163, 0.176) & 0.229(0.222, 0.235)\\
           \hline
    \multirow{5}{*}{\(J = 50\)}
    &Deep Survival Learner & 0.084(0.082, 0.086) & 0.074(0.072, 0.076) & 0.112(0.109, 0.115)\\
    &M-Learner             & 0.114(0.110, 0.117) & 0.122(0.117, 0.127) & 0.189(0.185, 0.194)\\
    &R-Learner             & 0.174(0.168, 0.180) & 0.237(0.229, 0.245) & 0.261(0.251, 0.270)\\
    &X-Learner             & 0.083(0.080, 0.085) & 0.094(0.090, 0.098) & 0.157(0.152, 0.161)\\
    &DRL-NFT-BART          & 0.077(0.062, 0.091) & 0.126(0.108, 0.143) & 0.087(0.081, 0.094)\\
    &PSDRL-NFT-BART        & 0.141(0.137, 0.146) & 0.152(0.147, 0.158) & 0.152(0.148, 0.156)\\
    &Causal Survival Forest& 0.156(0.152, 0.160) & 0.170(0.164, 0.177) & 0.228(0.222, 0.235)\\
             \hline
    \multirow{5}{*}{\(J = 100\)}
    &Deep Survival Learner & 0.084(0.082, 0.086) & 0.076(0.074, 0.077) & 0.114(0.111, 0.117)\\
    &M-Learner             & 0.114(0.110, 0.117) & 0.123(0.118, 0.128) & 0.190(0.185, 0.194)\\
    &R-Learner             & 0.173(0.167, 0.179) & 0.233(0.225, 0.241) & 0.260(0.252, 0.269)\\
    &X-Learner             & 0.082(0.080, 0.085) & 0.096(0.091, 0.100) & 0.154(0.150, 0.158)\\
    &DRL-NFT-BART          & 0.082(0.066, 0.098) & 0.139(0.117, 0.161) & 0.088(0.082, 0.094)\\
    &PSDRL-NFT-BART        & 0.142(0.138, 0.146) & 0.152(0.147, 0.158) & 0.150(0.146, 0.155)\\
    &Causal Survival Forest& 0.156(0.152, 0.161) & 0.171(0.164, 0.177) & 0.228(0.222, 0.234)\\
             \hline
        \end{tabular}}
     \begin{tablenotes}\footnotesize
\item \textsuperscript{1} Mean squared error and 95\% Monte Carlo confidence intervals are computed based on 200 simulation replicates.
\item \textsuperscript{2} Conditional average treatment effects are evaluated at \(J \in \{1, 50, 100\}\) time points, equally spaced between the 20th and 80th percentiles of the observed survival time distribution.
\item \textsuperscript{3} Columns correspond to combinations of correctly and incorrectly specified survival and treatment assignment models; the censoring model is correctly specified in all scenarios.
\end{tablenotes}
    \end{threeparttable}
\end{table}

 \section{Boston Lung Cancer Study}
\label{BLCS}
{ We apply the proposed method to investigate heterogeneous treatment effects of perioperative chemotherapy using data from the Boston Lung Cancer Study (BLCS), a large prospective longitudinal cancer cohort initiated in 1992 and followed through 2022~\citep{zhai2022spirometry,alvarez2025effects,wang2023prediagnosis}. The BLCS has enrolled over 12{,}000 patients treated at Massachusetts General Hospital and includes detailed clinical, pathological, and follow-up information, including treatment history and survival outcomes.}

{ We evaluate the effect of perioperative chemotherapy on overall survival among patients with non–small cell lung cancer (NSCLC). The primary outcome is the time from diagnosis to death or censoring. Because perioperative chemotherapy is primarily administered to patients with early-stage disease, we restrict the analysis to patients with stage II or III NSCLC who underwent surgical resection. Within this subset, we compare patients who received perioperative chemotherapy with those who did not. Patients with missing information on cancer histology, stage, treatment status, or diagnosis date are excluded. After applying these inclusion and exclusion criteria, a total of 784 patients from the BLCS cohort are included in the analysis.}

 Based on the descriptive statistics in Table~\ref{descriptive analysis}, the study cohort consists of 784 patients, among whom 521 (66\%) deaths were observed during follow-up. The median survival time for the overall population is 5.4 years (95\% CI: 4.7--6.1).
Patients who received surgery plus perioperative chemotherapy are defined as the treatment group, whereas those who underwent surgery alone constitute the control group. A total of 214 patients (27\%) received perioperative chemotherapy.
Baseline characteristics differ between the two groups. Patients in the treatment group are younger (median age 64 years, IQR: 58--71) than those in the control group (median age 68 years, IQR: 60--75), and a higher proportion of treated patients are male (54\% vs.\ 48\%).
Survival outcomes appear more favorable in the treatment group: 110 deaths (51\%) were observed among treated patients compared with 411 deaths (72\%) in the control group. The median survival time is 6.5 years (95\% CI: 5.3--8.6) for the treatment group, versus 4.8 years (95\% CI: 4.1--5.8) for the control group.
{ Missing values in age, BMI, and smoking intensity are imputed using their respective sample means (mean imputation).}  Kaplan--Meier curves (Figure~\ref{KM curve}) further illustrate differences in survival probabilities between patients who received perioperative chemotherapy and those who did not.

\begin{table}[hbpt]
    \caption{\textbf{Clinical Characteristics of Patients from the Boston Lung Cancer Cohort}}
    \label{descriptive analysis}
    \begin{threeparttable}
    \centering
    \fontsize{10}{12}\selectfont
    \makebox[\textwidth]{\begin{tabular}{lccc}
        \hline
        \textbf{Characteristic} & \multicolumn{1}{c}{\shortstack{\textbf{Overall}\\N = 784\textsuperscript{1}}} & \multicolumn{1}{c}{\shortstack{\textbf{\textbf{Surgery}}\\N = 570\textsuperscript{1}}} & \multicolumn{1}{c}{\shortstack{\textbf{\textbf{Surgery + Chemotherapy}}\\N = 214\textsuperscript{1}}} \\
        \hline
        Death & 521 (66\%) & 411 (72\%) & 110 (51\%) \\ 
        Median Survival Time (years) & 5.4 (4.7, 6.1) & 4.8 (4.1, 5.8) & 6.5 (5.3, 8.6) \\ 
        Age at disgonsis (yrs) & 67 (59, 74) & 68 (60, 75) & 64 (58, 71) \\ 
        \qquad    Unknown & 19 & 12 & 7 \\ 
        BMI (kg/m\textsuperscript{2}) & 25.8 (22.9, 29.5) & 25.8 (22.8, 29.6) & 26.3 (23.5, 29.5) \\ 
        \qquad  Unknown & 60 & 53 & 7 \\ 
        Tumor Stage &  &  &  \\ 
        \qquad    II & 445 (57\%) & 333 (58\%) & 112 (52\%) \\ 
        \qquad    III & 339 (43\%) & 237 (42\%) & 102 (48\%) \\ 
        Gender &  &  &  \\ 
        \qquad    Female & 394 (50\%) & 296 (52\%) & 98 (46\%) \\ 
        \qquad    Male & 390 (50\%) & 274 (48\%) & 116 (54\%) \\ 
        Race &  &  &  \\ 
        \qquad  White & 746 (95\%) & 547 (96\%) & 199 (93\%) \\ 
        \qquad  Asian & 15 (1.9\%) & 7 (1.2\%) & 8 (3.7\%) \\ 
        \qquad  Black & 8 (1.0\%) & 5 (0.9\%) & 3 (1.4\%) \\ 
        \qquad  Other & 15 (1.9\%) & 11 (1.9\%) & 4 (1.9\%) \\ 
        Smoking Status &  &  &  \\ 
        \qquad   Current Smoker & 213 (27\%) & 166 (29\%) & 47 (22\%) \\ 
        \qquad   Former Smoker & 490 (62\%) & 353 (62\%) & 137 (64\%) \\ 
        \qquad    Never Smoker & 81 (10\%) & 51 (8.9\%) & 30 (14\%) \\ 
              \hline
    \end{tabular}}
    \begin{tablenotes}\footnotesize
    \item[\textsuperscript{1}]Median (IQR); n (\%)
    \end{tablenotes}
    \end{threeparttable}
\end{table}

\begin{figure}[ht]
    \centering
    \includegraphics[scale=0.23]{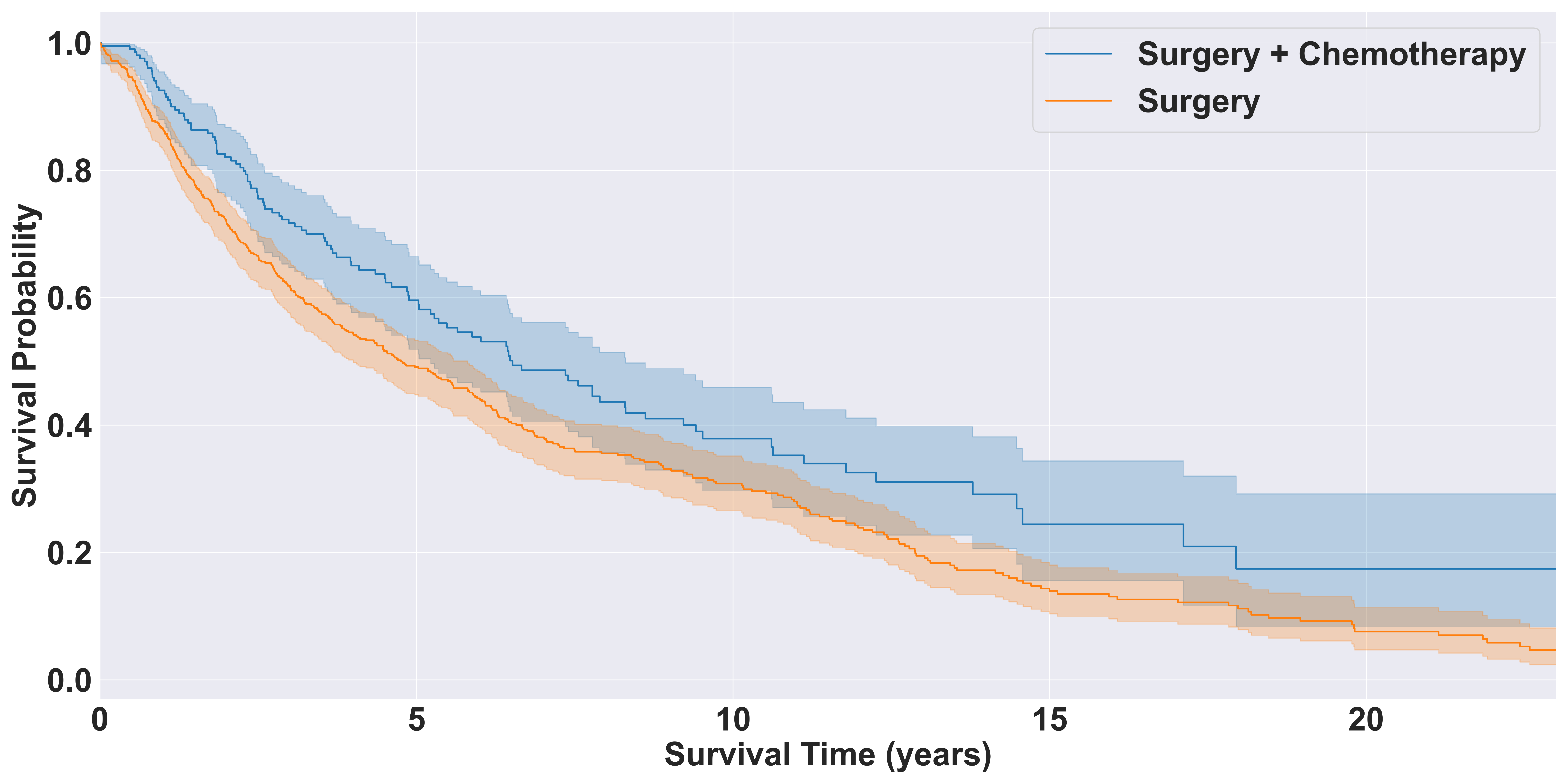}
    \caption{\textbf{Kaplan--Meier survival curves by treatment group.}} 
    \label{KM curve}
\end{figure}

 We estimate CATEs at 10 prespecified time points ranging from 1 to 23 years after surgery. At each time point $t$, the CATE is defined as the difference in survival probabilities at time $t$ between surgery plus perioperative chemotherapy and surgery alone, conditional on pre-treatment covariates.  { In our implementation, DSL employs a three-layer feedforward network with hidden layers of sizes 64, 32, and 16, respectively, and uses ReLU activation functions.}

 Figure~\ref{cate} presents the estimated CATEs across key baseline covariates, including gender, tumor stage, age, body mass index (BMI), and smoking intensity. Across all subgroups, perioperative chemotherapy is associated with positive treatment effects over the time horizon considered, corresponding to higher survival probabilities relative to surgery alone. The magnitude of the estimated CATE increases during the first ten years following surgery and attenuates thereafter.

{Gender-specific results (Figure~\ref{cate}(a)) indicate larger treatment effects for female patients. Holding other covariates at their mean or mode, the estimated increases in 5-, 10-, and 20-year survival probabilities are 3.9\%, 9.7\%, and 6.8\% for female patients, compared with 2.7\%, $-2.8\%$, and $-0.5\%$ for male patients. Substantial heterogeneity is also observed by tumor stage (Figure~\ref{cate}(b)). For patients with stage II NSCLC, perioperative chemotherapy is associated with increases of 3.9\%, 9.7\%, and 6.8\% in 5-, 10-, and 20-year survival probabilities, respectively. In contrast, the estimated effects for stage III patients are smaller, with corresponding increases of 3.3\%, 3.7\%, and 3.3\%.}

 Age-specific effects (Figure~\ref{cate}(c)) are relatively stable. Holding other covariates at their mean or mode, the estimated increases in 5-, 10-, and 20-year survival probabilities are 3.9\%, 10.1\%, and 7.1\% for patients aged 65 years, and 3.6\%, 6.7\%, and 5.1\% for patients aged 75 years. 
Figure~\ref{cate}(d) suggests that patients with higher BMI experience larger treatment effects: patients with BMI $=32$ show increases of 4.0\%, 10.1\%, and 7.1\% at 5, 10, and 20 years, compared with 3.8\%, 8.1\%, and 5.9\% for patients with BMI $=22$. 
Smoking intensity also modifies treatment effects (Figure~\ref{cate}(e)). Patients smoking 10 cigarettes per day exhibit larger survival gains (4.3\%, 13.2\%, and 8.9\%) than those smoking 30 cigarettes per day (3.7\%, 7.6\%, and 5.6\%).
Overall, perioperative chemotherapy is associated with meaningful survival benefits for patients with early-stage NSCLC, with effects concentrated in the first ten years following surgery and attenuating thereafter. This temporal pattern is consistent with long-term follow-up from the International Adjuvant Lung Cancer Trial, which reported a significant interaction between treatment and follow-up time, with benefits diminishing beyond five years~\citep{arriagada2010long}. 
The observed heterogeneity by gender and tumor stage aligns with prior studies reporting greater benefits among female patients~\citep{leiter2022benefits,sandler2016gender} and stronger effects in stage II disease relative to more advanced stages~\citep{douillard2006adjuvant}. The findings on smoking intensity are also consistent with prior evidence linking lower smoking exposure to improved chemotherapy outcomes~\citep{zhang2008influence}. 
Finally, patients with higher BMI appear to benefit more from perioperative chemotherapy, a pattern consistent with the reported obesity paradox in lung cancer~\citep{sakin2021effect}.
 These results underscore the importance of accounting for treatment effect heterogeneity and equity in evaluating treatment effectiveness~\citep{arriagada2010long,leiter2022benefits,sandler2016gender,douillard2006adjuvant,zhang2008influence}.
\begin{figure}[ht]
    \centering
    \includegraphics[scale=0.31]{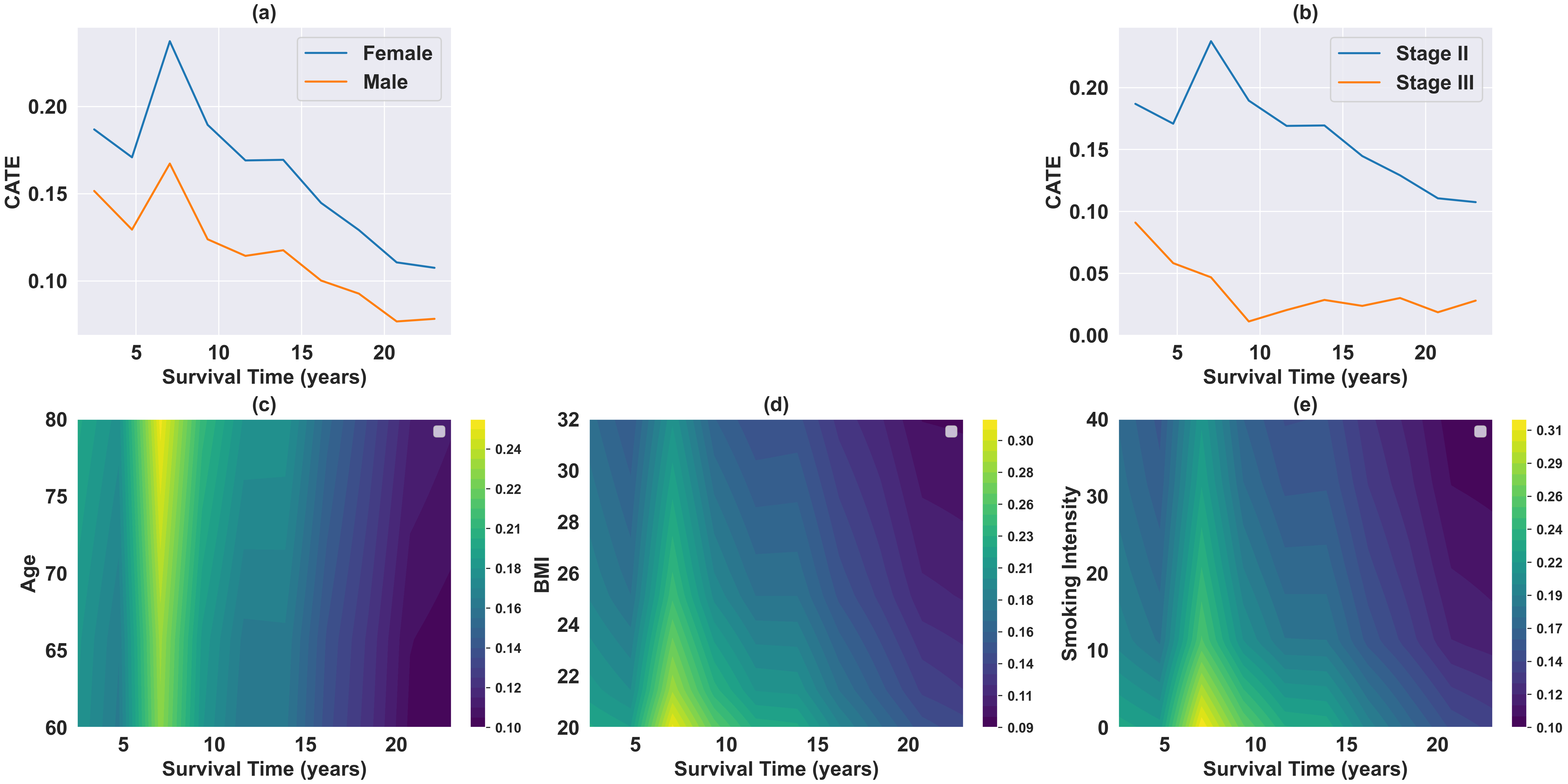}
        \caption{\textbf{Estimated conditional average treatment effects across key baseline covariates.} Each panel displays the estimated CATE as a function of one covariate of interest (gender, tumor stage, age, body mass index, or smoking intensity), with all remaining covariates fixed at their sample means or modal values.} 
    \label{cate}
\end{figure}

\section{Discussion}

 We propose the Deep Survival Learner (DSL), a causal deep learning framework for estimating heterogeneous treatment effects with right-censored survival outcomes. The key contribution is the integration of a doubly robust pseudo-outcome for time-specific conditional average treatment effects with joint learning across time via shared representations. This formulation directly targets treatment effect trajectories over clinically relevant follow-up intervals and improves efficiency relative to pointwise approaches that ignore temporal dependence.

The method relies on standard identification conditions, including unconfoundedness and appropriate handling of censoring. While the doubly robust construction mitigates misspecification of either the outcome or treatment model, correct specification of the censoring mechanism remains essential. This highlights the need for extensions that relax noninformative censoring, for example through joint modeling, sensitivity analysis, or alternative identification strategies.

From a statistical perspective, DSL performs joint estimation of a vector-valued target over a time grid, trading bias and variance through shared representations. This raises questions on optimal grid design, regularization, and approximation error in high-dimensional settings. Extensions to competing risks would further broaden the framework to multivariate event processes and require redefining target functionals.

Finally, our current analysis focuses on point estimation. Developing valid uncertainty quantification for time-indexed CATEs under censoring remains an important open problem, requiring new theoretical tools that account for dependence across time and the complexity of deep learners.

 \section*{Acknowledgements}
This work was supported by NIH grants 2R01CA249096-05 and R01CA269398. We are grateful to the Editor, the Associate Editor, and the two referees for their insightful reviews, which substantially improved the article. We also thank our long-term collaborator, Dr.~David Christiani, for providing the BLCS data and for his many insightful clinical suggestions.

\bibliographystyle{unsrtnat}
\bibliography{references}

\begin{thebibliography}{35}
\providecommand{\natexlab}[1]{#1}
\providecommand{\url}[1]{\texttt{#1}}
\expandafter\ifx\csname urlstyle\endcsname\relax
  \providecommand{\doi}[1]{doi: #1}\else
  \providecommand{\doi}{doi: \begingroup \urlstyle{rm}\Url}\fi

\bibitem[Imbens and Rubin(2016)]{imbens2016causal}
Guido~W Imbens and Donald~B Rubin.
\newblock \emph{Causal inference for statistics, social, and biomedical sciences: An introduction}.
\newblock Taylor \& Francis, 2016.

\bibitem[Dharmarajan et~al.(2017)Dharmarajan, Bragg-Gresham, Morgenstern, Gillespie, Li, Powe, Tuot, Banerjee, R{\'\i}os~Burrows, Rolka, Saydah, and Saran]{dharmarajan2017ckd}
Sai~H. Dharmarajan, Jennifer~L. Bragg-Gresham, Hal Morgenstern, Brenda~W. Gillespie, Yi~Li, Neil~R. Powe, Delphine~S. Tuot, Tanushree Banerjee, Nilka R{\'\i}os~Burrows, Deborah~B. Rolka, Sharon~H. Saydah, and Rajiv Saran.
\newblock State-level awareness of chronic kidney disease in the {US}.
\newblock \emph{American Journal of Preventive Medicine}, 53\penalty0 (3):\penalty0 300--307, 2017.
\newblock \doi{10.1016/j.amepre.2017.03.023}.

\bibitem[Haddad et~al.(2008)Haddad, Crum, Chen, Krane, Posner, Li, and Burk]{haddad2008hpv16}
Robert Haddad, Christopher Crum, Zigui Chen, Jeffrey Krane, Marshall Posner, Yi~Li, and Robert Burk.
\newblock {HPV16} transmission between a couple with {HPV}-related head and neck cancer.
\newblock \emph{Oral Oncology}, 44\penalty0 (8):\penalty0 812--815, 2008.
\newblock \doi{10.1016/j.oraloncology.2007.09.004}.

\bibitem[Zhai et~al.(2022)Zhai, Li, Brown, Lanuti, Gainor, and Christiani]{zhai2022spirometry}
Ting Zhai, Yi~Li, Robert Brown, Michael Lanuti, Justin~F Gainor, and David~C Christiani.
\newblock Spirometry at diagnosis and overall survival in non-small cell lung cancer patients.
\newblock \emph{Cancer Medicine}, 11\penalty0 (24):\penalty0 4796--4805, 2022.

\bibitem[Horvitz and Thompson(1952)]{horvitz1952generalization}
Daniel~G Horvitz and Donovan~J Thompson.
\newblock A generalization of sampling without replacement from a finite universe.
\newblock \emph{Journal of the American Statistical Association}, 47\penalty0 (260):\penalty0 663--685, 1952.

\bibitem[K{\"u}nzel et~al.(2019)K{\"u}nzel, Sekhon, Bickel, and Yu]{kunzel2019metalearners}
S{\"o}ren~R K{\"u}nzel, Jasjeet~S Sekhon, Peter~J Bickel, and Bin Yu.
\newblock Metalearners for estimating heterogeneous treatment effects using machine learning.
\newblock \emph{Proceedings of the National academy of Sciences}, 116\penalty0 (10):\penalty0 4156--4165, 2019.

\bibitem[Nie and Wager(2021)]{nie2021quasi}
Xinkun Nie and Stefan Wager.
\newblock Quasi-oracle estimation of heterogeneous treatment effects.
\newblock \emph{Biometrika}, 108\penalty0 (2):\penalty0 299--319, 2021.

\bibitem[Kennedy(2020)]{kennedy2020towards}
Edward~H Kennedy.
\newblock Towards optimal doubly robust estimation of heterogeneous causal effects.
\newblock \emph{arXiv preprint arXiv:2004.14497}, 2020.

\bibitem[Yang et~al.(2025)Yang, Hu, Liao, and Chen]{yang2025doubly}
Yuhui Yang, Weiwei Hu, Zhenli Liao, and Fangyao Chen.
\newblock Doubly robust estimators for heterogeneous treatment effects in heteroskedastic survival data.
\newblock \emph{Statistics in Medicine}, 44\penalty0 (23-24):\penalty0 e70301, 2025.

\bibitem[Bo et~al.(2025)Bo, Jeong, Forno, and Ding]{bo2025evaluating}
Na~Bo, Jong-Hyeon Jeong, Erick Forno, and Ying Ding.
\newblock Evaluating meta-learners to analyze treatment heterogeneity in survival data: Application to electronic health records of pediatric asthma care in covid-19 pandemic.
\newblock \emph{Statistics in Medicine}, 44\penalty0 (3-4):\penalty0 e10333, 2025.

\bibitem[Cui et~al.(2023)Cui, Kosorok, Sverdrup, Wager, and Zhu]{cui2023estimating}
Yifan Cui, Michael~R Kosorok, Erik Sverdrup, Stefan Wager, and Ruoqing Zhu.
\newblock Estimating heterogeneous treatment effects with right-censored data via causal survival forests.
\newblock \emph{Journal of the Royal Statistical Society Series B: Statistical Methodology}, 85\penalty0 (2):\penalty0 179--211, 2023.

\bibitem[Dandl et~al.(2024)Dandl, Bender, and Hothorn]{dandl2024heterogeneous}
Susanne Dandl, Andreas Bender, and Torsten Hothorn.
\newblock Heterogeneous treatment effect estimation for observational data using model-based forests.
\newblock \emph{Statistical Methods in Medical Research}, 33\penalty0 (3):\penalty0 392--413, 2024.

\bibitem[Hu(2024)]{hu2024new}
Liangyuan Hu.
\newblock A new method for clustered survival data: Estimation of treatment effect heterogeneity and variable selection.
\newblock \emph{Biometrical Journal}, 66\penalty0 (1):\penalty0 2200178, 2024.

\bibitem[Curth et~al.(2021)Curth, Lee, and van~der Schaar]{curth2021survite}
Alicia Curth, Changhee Lee, and Mihaela van~der Schaar.
\newblock Survite: Learning heterogeneous treatment effects from time-to-event data.
\newblock \emph{Advances in Neural Information Processing Systems}, 34:\penalty0 26740--26753, 2021.

\bibitem[Xu et~al.(2023)Xu, Ignatiadis, Sverdrup, Fleming, Wager, and Shah]{xu2023treatment}
Yizhe Xu, Nikolaos Ignatiadis, Erik Sverdrup, Scott Fleming, Stefan Wager, and Nigam Shah.
\newblock Treatment heterogeneity with survival outcomes.
\newblock In \emph{Handbook of matching and weighting adjustments for causal inference}, pages 445--482. Chapman and Hall/CRC, 2023.

\bibitem[Rosenbaum and Rubin(1983)]{rosenbaum1983propensity}
Paul~R Rosenbaum and Donald~B Rubin.
\newblock The central role of the propensity score in observational studies for causal effects.
\newblock \emph{Biometrika}, 70\penalty0 (1):\penalty0 41--55, 1983.

\bibitem[Imbens and Rubin(2015)]{imbens2015causal}
Guido~W Imbens and Donald~B Rubin.
\newblock \emph{Causal Inference for Statistics, Social, and Biomedical Sciences: An Introduction}.
\newblock Cambridge University Press, 2015.

\bibitem[Andersen and Gill(1982)]{andersen1982cox}
Per~Kragh Andersen and Richard~D Gill.
\newblock Cox's regression model for counting processes: A large sample study.
\newblock \emph{Annals of Statistics}, 10\penalty0 (4):\penalty0 1100--1120, 1982.

\bibitem[Tsiatis(2006)]{tsiatis2006semiparametric}
Anastasios~A Tsiatis.
\newblock \emph{Semiparametric Theory and Missing Data}.
\newblock Springer, 2006.

\bibitem[van~der Vaart(1998)]{vanderVaart1998}
Aad~W van~der Vaart.
\newblock \emph{Asymptotic Statistics}.
\newblock Cambridge University Press, 1998.

\bibitem[Chernozhukov et~al.(2018{\natexlab{a}})Chernozhukov, Chetverikov, Demirer, Duflo, Hansen, Newey, and Robins]{chernozhukov2018dml}
Victor Chernozhukov, Denis Chetverikov, Mert Demirer, Esther Duflo, Christian Hansen, Whitney Newey, and James Robins.
\newblock Double/debiased machine learning for treatment and structural parameters.
\newblock \emph{The Econometrics Journal}, 21\penalty0 (1):\penalty0 C1--C68, 2018{\natexlab{a}}.

\bibitem[Tsybakov(2009)]{tsybakov2009introduction}
Alexandre~B. Tsybakov.
\newblock \emph{Introduction to Nonparametric Estimation}.
\newblock Springer Series in Statistics. Springer, 2009.
\newblock \doi{10.1007/b13794}.

\bibitem[Gy{\"o}rfi et~al.(2006)Gy{\"o}rfi, Kohler, Krzy{\.z}ak, and Walk]{gyorfi2006distribution}
L{\'a}szl{\'o} Gy{\"o}rfi, Michael Kohler, Adam Krzy{\.z}ak, and Harro Walk.
\newblock \emph{A Distribution-Free Theory of Nonparametric Regression}.
\newblock Springer, 2006.

\bibitem[Chernozhukov et~al.(2018{\natexlab{b}})Chernozhukov, Chetverikov, Demirer, Duflo, Hansen, Newey, and Robins]{chernozhukov2018double}
Victor Chernozhukov, Denis Chetverikov, Mert Demirer, Esther Duflo, Christian Hansen, Whitney Newey, and James Robins.
\newblock Double/debiased machine learning for treatment and structural parameters.
\newblock \emph{The Econometrics Journal}, 21\penalty0 (1):\penalty0 C1--C68, 2018{\natexlab{b}}.
\newblock \doi{10.1111/ectj.12097}.

\bibitem[Schmidt-Hieber(2020)]{schmidt2020deep}
Johannes Schmidt-Hieber.
\newblock Deep relu network approximation of functions on a manifold.
\newblock \emph{Journal of Machine Learning Research}, 21\penalty0 (52):\penalty0 1--26, 2020.

\bibitem[van~der Vaart and Wellner(1996)]{van1996weak}
Aad~W. van~der Vaart and Jon~A. Wellner.
\newblock \emph{Weak Convergence and Empirical Processes}.
\newblock Springer, 1996.

\bibitem[Bartlett et~al.(2017)Bartlett, Foster, and Telgarsky]{bartlett2017spectrally}
Peter~L. Bartlett, Dylan~J. Foster, and Matus Telgarsky.
\newblock Spectrally-normalized margin bounds for neural networks.
\newblock In \emph{Advances in Neural Information Processing Systems}, volume~30, 2017.

\bibitem[Alvarez et~al.(2026)Alvarez, Sun, Li, and Christiani]{alvarez2025effects}
A.~A.~R. Alvarez, Y.~Sun, Y.~Li, and D.~C. Christiani.
\newblock Effects of sex on mortality in patients with lung cancer: A multiple mediation analysis of the {Boston Lung Cancer Study}.
\newblock \emph{Clinical Lung Cancer}, 27\penalty0 (2):\penalty0 201--209.e3, March 2026.

\bibitem[Wang et~al.(2023)Wang, Romero-Gutierrez, Kothari, Shafer, Li, and Christiani]{wang2023prediagnosis}
Xinan Wang, Christopher~W Romero-Gutierrez, Jui Kothari, Andrea Shafer, Yi~Li, and David~C Christiani.
\newblock Prediagnosis smoking cessation and overall survival among patients with non--small cell lung cancer.
\newblock \emph{JAMA Network Open}, 6\penalty0 (5):\penalty0 e2311966, 2023.

\bibitem[Arriagada et~al.(2010)Arriagada, Dunant, Pignon, Bergman, Chabowski, Grunenwald, Kozlowski, Le~P{\'e}choux, Pirker, Pinel, et~al.]{arriagada2010long}
Rodrigo Arriagada, Ariane Dunant, Jean-Pierre Pignon, Bengt Bergman, Mariusz Chabowski, Dominique Grunenwald, Miroslaw Kozlowski, C{\'e}cile Le~P{\'e}choux, Robert Pirker, Maria-Izabel~Sathler Pinel, et~al.
\newblock Long-term results of the international adjuvant lung cancer trial evaluating adjuvant cisplatin-based chemotherapy in resected lung cancer.
\newblock \emph{Journal of Clinical Oncology}, 28\penalty0 (1):\penalty0 35--42, 2010.

\bibitem[Leiter et~al.(2022)Leiter, Kong, Gould, Kale, Veluswamy, Smith, Mhango, Huang, Wisnivesky, and Sigel]{leiter2022benefits}
Amanda Leiter, Chung~Yin Kong, Michael~K Gould, Minal~S Kale, Rajwanth~R Veluswamy, Cardinale~B Smith, Grace Mhango, Brian~Z Huang, Juan~P Wisnivesky, and Keith Sigel.
\newblock The benefits and harms of adjuvant chemotherapy for non-small cell lung cancer in patients with major comorbidities: A simulation study.
\newblock \emph{Plos one}, 17\penalty0 (11):\penalty0 e0263911, 2022.

\bibitem[Sandler et~al.(2016)Sandler, Wang, Hancock, Boffa, Detterbeck, and Kim]{sandler2016gender}
Britt~J Sandler, Zuoheng Wang, Jacquelyn~G Hancock, Daniel~J Boffa, Frank~C Detterbeck, and Anthony~W Kim.
\newblock Gender, age, and comorbidity status predict improved survival with adjuvant chemotherapy following lobectomy for non-small cell lung cancers larger than 4 cm.
\newblock \emph{Annals of surgical oncology}, 23:\penalty0 638--645, 2016.

\bibitem[Douillard et~al.(2006)Douillard, Rosell, De~Lena, Carpagnano, Ramlau, Gonz{\'a}les-Larriba, Grodzki, Pereira, Le~Groumellec, Lorusso, et~al.]{douillard2006adjuvant}
Jean-Yves Douillard, Rafael Rosell, Mario De~Lena, Francesco Carpagnano, Rodryg Ramlau, Jose~Luis Gonz{\'a}les-Larriba, Tomasz Grodzki, Jose~Rodrigues Pereira, Alain Le~Groumellec, Vito Lorusso, et~al.
\newblock Adjuvant vinorelbine plus cisplatin versus observation in patients with completely resected stage ib--iiia non-small-cell lung cancer (adjuvant navelbine international trialist association [anita]): a randomised controlled trial.
\newblock \emph{The Lancet Oncology}, 7\penalty0 (9):\penalty0 719--727, 2006.

\bibitem[Zhang et~al.(2008)Zhang, Xu, Wang, Li, and Wang]{zhang2008influence}
Zhenfa Zhang, Feng Xu, Shengguang Wang, Ni~Li, and Changli Wang.
\newblock Influence of smoking on histologic type and the efficacy of adjuvant chemotherapy in resected non-small cell lung cancer.
\newblock \emph{Lung Cancer}, 60\penalty0 (3):\penalty0 434--440, 2008.

\bibitem[Sakin et~al.(2021)Sakin, Sahin, Mustafa~Atci, Yasar, Demir, Geredeli, Sakin, and Cihan]{sakin2021effect}
Aysegul Sakin, Suleyman Sahin, Muhammed Mustafa~Atci, Nurgul Yasar, Cumhur Demir, Caglayan Geredeli, Abdullah Sakin, and Sener Cihan.
\newblock The effect of body mass index on treatment outcomes in patients with metastatic non-small cell lung cancer treated with platinum-based therapy.
\newblock \emph{Nutrition and Cancer}, 73\penalty0 (8):\penalty0 1411--1418, 2021.

\end{thebibliography}

\clearpage
\setcounter{page}{1}

\renewcommand{\thetable}{A.\arabic{table}}
\setcounter{table}{0}

\newpage 
\section*{Supplementary Materials for \\ ``estimating heterogeneous treatment effects with survival outcomes via a deep survival learner" }
\label{app1}

\noindent
{\bf Proofs of Theorems}

 \begin{proofof} {Theorem~\ref{thm:ID}:}
It suffices to show that, for each $w \in \{0,1\}$,
\[
\mathbb{E}\{\pi_i^w(t) \mid X_i = x\} = P(T_i^w > t \mid X_i = x),
\]
since this immediately implies
\[
\mathbb{E}\{\varphi_i(t) \mid X_i = x\}
=
P(T_i^1 > t \mid X_i = x) - P(T_i^0 > t \mid X_i = x)
=
\tau_0(x;t).
\]
We prove the result for $w = 1$; the case $w = 0$ follows by symmetry.

Recall that
\[
\pi_i^1(t)
=
S^1(x;t)
+
\frac{\mathbb{I}(W_i = 1)}{e^1(x)}
\left\{
\frac{\mathbb{I}(U_i > t)}{G^1(x;t)}
-
S^1(x;t)
\right\},
\]
where $S^1(x;t) = P(T_i > t \mid X_i = x, W_i = 1)$,
$e^1(x) = P(W_i = 1 \mid X_i = x)$, and
$G^1(x;t) = P(\mathcal C_i > t \mid X_i = x, W_i = 1)$.

Taking conditional expectation given $X_i = x$ yields
\begin{align*}
\mathbb{E}\{\pi_i^1(t) \mid X_i = x\}
&=
S^1(x;t)
+
\mathbb{E}\!\left[
\left.
\frac{\mathbb{I}(W_i = 1)}{e^1(x)}
\left\{
\frac{\mathbb{I}(U_i > t)}{G^1(x;t)}
-
S^1(x;t)
\right\}
\right| X_i = x
\right].
\end{align*}

By the law of iterated expectations, the second term can be written as
\begin{align*}
&\mathbb{E}\!\left[
\left.
\mathbb{E}\!\left\{
\left.
\frac{\mathbb{I}(W_i = 1)}{e^1(x)}
\left(
\frac{\mathbb{I}(U_i > t)}{G^1(x;t)}
-
S^1(x;t)
\right)
\right| X_i = x, W_i
\right\}
\right| X_i = x
\right].
\end{align*}

Conditioning on $W_i$, and noting that $\mathbb{I}(W_i = 1)$ is deterministic given $W_i$, we obtain
\begin{align*}
&=
\mathbb{E}\!\left[
\left.
\mathbb{I}(W_i = 1)
\frac{1}{e^1(x)}
\left\{
\mathbb{E}\!\left[
\left.
\frac{\mathbb{I}(U_i > t)}{G^1(x;t)}
\right| X_i = x, W_i = 1
\right]
-
S^1(x;t)
\right\}
\right| X_i = x
\right].
\end{align*}

Using the definition $U_i = \min(T_i, \mathcal C_i)$, we have
$\mathbb{I}(U_i > t) = \mathbb{I}(T_i > t,\, \mathcal C_i > t)$.
Under Assumption~\ref{assum_indep_censoring},
 $\mathcal C_i \perp T_i^w \mid (X_i, W_i=w)$ for $w\in\{0,1\}$.
Under Assumption~\ref{assum_consistency}, when conditioning on $W_i=w$ we have
$T_i = T_i^w$, which implies $\mathcal C_i \perp T_i \mid (X_i, W_i=w)$, and  in particular,
$T_i \perp \mathcal C_i \mid (X_i, W_i = 1)$. This implies
\[
\mathbb{E}\!\left[
\left.
\mathbb{I}(T_i > t,\, \mathcal C_i > t)
\right| X_i = x, W_i = 1
\right]
=
P(T_i > t \mid X_i = x, W_i = 1)
P(\mathcal C_i > t \mid X_i = x, W_i = 1).
\]
Therefore,
\[
\mathbb{E}\!\left[
\left.
\frac{\mathbb{I}(U_i > t)}{G^1(x;t)}
\right| X_i = x, W_i = 1
\right]
=
\frac{
P(T_i > t \mid X_i = x, W_i = 1)
P(\mathcal C_i > t \mid X_i = x, W_i = 1)
}{
P(\mathcal C_i > t \mid X_i = x, W_i = 1)
}
=
S^1(x;t).
\]

Substituting back yields
\begin{align*}
\mathbb{E}\{\pi_i^1(t) \mid X_i = x\}
&=
S^1(x;t)
+
\mathbb{E}\!\left[
\left.
\mathbb{I}(W_i = 1)
\frac{1}{e^1(x)}
\bigl\{ S^1(x;t) - S^1(x;t) \bigr\}
\right| X_i = x
\right] \\
&=
S^1(x;t).
\end{align*}

Finally, by Assumptions~\ref{assum_consistency} and~\ref{assum_unconfoundedness},
\[
S^1(x;t)
=
P(T_i > t \mid X_i = x, W_i = 1)
=
P(T_i^1 > t \mid X_i = x).
\]
This completes the proof.
\end{proofof}

 \begin{proofof} {Theorem~\ref{thm:DR}:}
It suffices to show that for each $w\in\{0,1\}$,
\[
\mathbb{E}\{\pi_i^w(t)\mid X_i=x\}=P(T_i^w>t\mid X_i=x),
\]
because then
\(
\mathbb{E}\{\varphi_i(t)\mid X_i=x\}
=
P(T_i^1>t\mid X_i=x)-P(T_i^0>t\mid X_i=x)
=
\tau_0(x;t).
\)
We prove the claim for $w=1$; the case $w=0$ follows by symmetry.

Recall the definition
\[
\pi_i^1(t)
=
S^1(X_i;t)
+
\frac{\mathbb{I}(W_i=1)}{e^1(X_i)}
\left\{
\frac{Y_i(t)}{G^1(X_i;t)}-S^1(X_i;t)
\right\},
\]
where $Y_i(t)=\mathbb{I}(U_i>t)=\mathbb{I}(T_i>t,\,\mathcal C_i>t)$, and where the nuisance functions
$S^1(x;t)$, $G^1(x;t)$ and $e^1(x)$ are intended to approximate
\[
S_0^1(x;t)=P(T_i>t\mid X_i=x,W_i=1),\qquad
G_0^1(x;t)=P(\mathcal C_i>t\mid X_i=x,W_i=1),
\]
\[
e_0^1(x)=P(W_i=1\mid X_i=x).
\]
Throughout, Assumptions 1--5 ensure that these conditional probabilities are well-defined and that $G_0^1(x;t)>0$.

\medskip

\noindent \textbf{Case A: the survival model is correctly specified.}
Assume $S^1(x;t)=S_0^1(x;t)$ and the censoring model is correctly specified, i.e., $G^1(x;t)=G_0^1(x;t)$.
Then, regardless of whether the propensity score model $e^1(\cdot)$ is correct,
the argument in Theorem~1 applies verbatim and yields
\[
\mathbb{E}\{\pi_i^1(t)\mid X_i=x\}=S_0^1(x;t)=P(T_i^1>t\mid X_i=x),
\]
where the last equality follows from Assumptions \ref{assum_consistency}  and \ref{assum_unconfoundedness}.
For completeness, we provide the key step:
\[
\mathbb{E}\!\left[\left.\frac{Y_i(t)}{G_0^1(x;t)}\right|X_i=x,W_i=1\right]
=
\frac{P(T_i>t,\mathcal C_i>t\mid X_i=x,W_i=1)}{P(\mathcal C_i>t\mid X_i=x,W_i=1)}
=
P(T_i>t\mid X_i=x,W_i=1),
\]
where the second equality uses Assumption 4 (noninformative censoring), i.e.,
$T_i\perp \mathcal C_i\mid (X_i,W_i=1)$.
This establishes unbiasedness in Case A.

\medskip

\noindent \textbf{Case B: the propensity score model is correctly specified.}
Now assume the censoring model is correctly specified, $G^1(x;t)=G_0^1(x;t)$, and the propensity score is correctly specified,
$e^1(x)=e_0^1(x)=P(W_i=1\mid X_i=x)$, while allowing the survival model $S^1(x;t)$ to be misspecified.
We show that $\mathbb{E}\{\pi_i^1(t)\mid X_i=x\}=S_0^1(x;t)$.

Condition on $X_i=x$. Using the definition of $\pi_i^1(t)$,
\begin{align*}
\mathbb{E}\{\pi_i^1(t)\mid X_i=x\}
&=
S^1(x;t)
+
\mathbb{E}\!\left[
\left.
\frac{\mathbb{I}(W_i=1)}{e_0^1(x)}
\left\{
\frac{Y_i(t)}{G_0^1(x;t)}-S^1(x;t)
\right\}
\right|X_i=x
\right].
\end{align*}
Apply the law of iterated expectations conditioning on $(X_i,W_i)$:
\begin{align*}
&\mathbb{E}\!\left[
\left.
\frac{\mathbb{I}(W_i=1)}{e_0^1(x)}
\left\{
\frac{Y_i(t)}{G_0^1(x;t)}-S^1(x;t)
\right\}
\right|X_i=x
\right]\\
&\qquad=
\mathbb{E}\!\left[
\left.
\mathbb{E}\!\left\{
\left.
\frac{\mathbb{I}(W_i=1)}{e_0^1(x)}
\left(
\frac{Y_i(t)}{G_0^1(x;t)}-S^1(x;t)
\right)
\right|X_i=x,W_i
\right\}
\right|X_i=x
\right].
\end{align*}
Inside the inner expectation, $\mathbb{I}(W_i=1)$ is fixed given $W_i$. Hence
\begin{align*}
&=
\mathbb{E}\!\left[
\left.
\frac{\mathbb{I}(W_i=1)}{e_0^1(x)}
\left\{
\mathbb{E}\!\left[\left.\frac{Y_i(t)}{G_0^1(x;t)}\right|X_i=x,W_i=1\right]
-
S^1(x;t)
\right\}
\right|X_i=x
\right].
\end{align*}
We next compute the conditional mean term. Since $Y_i(t)=\mathbb{I}(T_i>t,\mathcal C_i>t)$ and $G_0^1(x;t)=P(\mathcal C_i>t\mid X_i=x,W_i=1)$,
\begin{align*}
\mathbb{E}\!\left[\left.\frac{Y_i(t)}{G_0^1(x;t)}\right|X_i=x,W_i=1\right]
&=
\frac{
\mathbb{E}\!\left[\mathbb{I}(T_i>t,\mathcal C_i>t)\mid X_i=x,W_i=1\right]
}{
P(\mathcal C_i>t\mid X_i=x,W_i=1)
}.
\end{align*}
By Assumption \ref{assum_indep_censoring}   and a similar argument in the Proof of Theorem 1,
$T_i\perp \mathcal C_i\mid (X_i,W_i=1)$, so
\[
P(T_i>t,\mathcal C_i>t\mid X_i=x,W_i=1)
=
P(T_i>t\mid X_i=x,W_i=1)\,P(\mathcal C_i>t\mid X_i=x,W_i=1).
\]
Therefore,
\begin{align*}
\mathbb{E}\!\left[\left.\frac{Y_i(t)}{G_0^1(x;t)}\right|X_i=x,W_i=1\right]
&=
P(T_i>t\mid X_i=x,W_i=1)
=
S_0^1(x;t).
\end{align*}
Substituting this back,
\begin{align*}
\mathbb{E}\{\pi_i^1(t)\mid X_i=x\}
&=
S^1(x;t)
+
\mathbb{E}\!\left[
\left.
\frac{\mathbb{I}(W_i=1)}{e_0^1(x)}
\left\{
S_0^1(x;t)-S^1(x;t)
\right\}
\right|X_i=x
\right]\\
&=
S^1(x;t)
+
\left\{
S_0^1(x;t)-S^1(x;t)
\right\}
\mathbb{E}\!\left[\left.\frac{\mathbb{I}(W_i=1)}{e_0^1(x)}\right|X_i=x\right].
\end{align*}
Because $e_0^1(x)=P(W_i=1\mid X_i=x)$, we have
\[
\mathbb{E}\!\left[\left.\frac{\mathbb{I}(W_i=1)}{e_0^1(x)}\right|X_i=x\right]
=
\frac{\mathbb{E}\{\mathbb{I}(W_i=1)\mid X_i=x\}}{e_0^1(x)}
=
\frac{P(W_i=1\mid X_i=x)}{P(W_i=1\mid X_i=x)}
=
1.
\]
This identity follows from correct specification of the propensity score
and Assumption~\ref{assum_overlap}, which guarantees $0<e_0^1(x)<1$.
Hence,
\[
\mathbb{E}\{\pi_i^1(t)\mid X_i=x\}=S_0^1(x;t)=P(T_i>t\mid X_i=x,W_i=1).
\]
Finally, by Assumptions \ref{assum_consistency} and \ref{assum_unconfoundedness}, 
\[
P(T_i>t\mid X_i=x,W_i=1)=P(T_i^1>t\mid X_i=x),
\]
so $\mathbb{E}\{\pi_i^1(t)\mid X_i=x\}=P(T_i^1>t\mid X_i=x)$, completing the proof for $w=1$.

\medskip

\noindent The proof for $w=0$ is identical after replacing $1$ by $0$ throughout.
\end{proofof}
\begin{proofof} {Theorem~\ref{thm:bias_full}:}
    We first analyze the treatment-specific component \(\hat\pi_i^w(t)\) for fixed \(w\in\{0,1\}\).  As defined in the notation section,
    conditioning on \(X_i=x\), and using cross-fitting so that \(\hat S^w(x;t)\), \(\hat e^w(x)\), and \(\hat G^w(x;t)\) are fixed, we take expectation with respect to  \(T_i, \mathcal C_i\):
\begin{align*}
\mathbb E\{\hat\pi_i^w(t)\mid X_i=x\}
&=
\hat S^w(x;t)
+
\frac{1}{\hat e^w(x)}
\mathbb E\!\left[
\left.
\mathbb I(W_i=w)
\left\{
\frac{Y_i(t)}{\hat G^w(x;t)}
-
\hat S^w(x;t)
\right\}
\right|X_i=x
\right].
\end{align*}
By iterated expectation,
\begin{align*}
&\mathbb E\!\left[
\left.
\mathbb I(W_i=w)
\left\{
\frac{Y_i(t)}{\hat G^w(x;t)}
-
\hat S^w(x;t)
\right\}
\right|X_i=x
\right]\\
&\qquad=
e_0^w(x)
\left\{
\mathbb E\!\left[
\left.\frac{Y_i(t)}{\hat G^w(x;t)}\right|X_i=x,W_i=w
\right]
-
\hat S^w(x;t)
\right\}.
\end{align*}
Since \(Y_i(t)=\mathbb I(T_i>t,\mathcal C_i>t)\), by noninformative censoring,
\begin{align*}
\mathbb E\!\left[
\left.\frac{Y_i(t)}{\hat G^w(x;t)}\right|X_i=x,W_i=w
\right]
&=
\frac{P(T_i>t,\mathcal C_i>t\mid X_i=x,W_i=w)}{\hat G^w(x;t)}\\
&=
\frac{P(T_i>t\mid X_i=x,W_i=w)\,P(\mathcal C_i>t\mid X_i=x,W_i=w)}{\hat G^w(x;t)}\\
&=
\frac{S_0^w(x;t)\,G_0^w(x;t)}{\hat G^w(x;t)}.
\end{align*}
Hence,
\[
\mathbb E\{\hat\pi_i^w(t)\mid X_i=x\}
=
\hat S^w(x;t)
+
\frac{e_0^w(x)}{\hat e^w(x)}
\left\{
\frac{S_0^w(x;t)\,G_0^w(x;t)}{\hat G^w(x;t)}
-
\hat S^w(x;t)
\right\}.
\]
Subtracting \(S_0^w(x;t)\) from both sides gives
\begin{align*}
&\mathbb E\{\hat\pi_i^w(t)\mid X_i=x\} - S_0^w(x;t)\\
&\qquad=
\hat S^w(x;t)-S_0^w(x;t)
+
\frac{e_0^w(x)}{\hat e^w(x)}
\left\{
\frac{S_0^w(x;t)\,G_0^w(x;t)}{\hat G^w(x;t)}
-
\hat S^w(x;t)
\right\}.
\end{align*}
Add and subtract \(\frac{e_0^w(x)}{\hat e^w(x)}S_0^w(x;t)\) to obtain
\begin{align*}
&\mathbb E\{\hat\pi_i^w(t)\mid X_i=x\} - S_0^w(x;t)\\
&\qquad=
\left\{1-\frac{e_0^w(x)}{\hat e^w(x)}\right\}
\big\{\hat S^w(x;t)-S_0^w(x;t)\big\}
+
\frac{e_0^w(x)}{\hat e^w(x)}
\left\{
\frac{S_0^w(x;t)\,G_0^w(x;t)}{\hat G^w(x;t)}
-
S_0^w(x;t)
\right\}\\
&\qquad=
\frac{\hat e^w(x)-e_0^w(x)}{\hat e^w(x)}
\big\{\hat S^w(x;t)-S_0^w(x;t)\big\}
+
\frac{e_0^w(x)}{\hat e^w(x)}
\frac{S_0^w(x;t)}{\hat G^w(x;t)}
\big\{G_0^w(x;t)-\hat G^w(x;t)\big\}.
\end{align*}

Now use
\[
\hat\varphi_i(t)=\hat\pi_i^1(t)-\hat\pi_i^0(t),
\qquad
\varphi_i^0(t)=\pi_{i,0}^1(t)-\pi_{i,0}^0(t),
\]
together with
\[
\mathbb E\{\pi_{i,0}^w(t)\mid X_i=x\}=S_0^w(x;t),
\qquad w=0,1,
\]
to obtain
\begin{align*}
\mathbb E\big[\hat\varphi_i(t)-\varphi_i^0(t)\mid X_i=x\big]
&=
\Big(\mathbb E\{\hat\pi_i^1(t)\mid X_i=x\}-S_0^1(x;t)\Big)\\
&\qquad-
\Big(\mathbb E\{\hat\pi_i^0(t)\mid X_i=x\}-S_0^0(x;t)\Big),
\end{align*}
which yields \eqref{eq:full_bias_decomp}.

It remains to bound the right-hand side uniformly in \(x\). By positivity and uniform consistency of \(\hat e^w\) and \(\hat G^w\) (Assumption \ref{assum_uni_consistency}) , with probability tending to one,
\[
\inf_x \hat e^w(x)\ge c/2,
\qquad
\inf_x \hat G^w(x;t)\ge c/2.
\]
Also \(0\le S_0^w(x;t)\le 1\). Therefore,
\begin{align*}
&\sup_x
\left|
\frac{\hat e^w(x)-e_0^w(x)}{\hat e^w(x)}
\big\{\hat S^w(x;t)-S_0^w(x;t)\big\}
\right|
=
O_p\!\left(r_n^{(e)}r_n^{(S)}\right),
\\
&\sup_x
\left|
\frac{e_0^w(x)}{\hat e^w(x)}
\frac{S_0^w(x;t)}{\hat G^w(x;t)}
\big\{G_0^w(x;t)-\hat G^w(x;t)\big\}
\right|
=
O_p\!\left(r_n^{(G)}\right).
\end{align*}
Summing over \(w=0,1\) yields
\[
\sup_x
\left|
\mathbb E\big[\hat\varphi_i(t)-\varphi_i^0(t)\mid X_i=x\big]
\right|
=
O_p\!\left(r_n^{(e)}r_n^{(S)}+r_n^{(G)}\right).
\]
If \(\hat G^w=G_0^w\), the second term vanishes and the bound reduces to
\[
O_p\!\left(r_n^{(e)}r_n^{(S)}\right).
\]
This completes the proof.
\end{proofof}

\begin{proofof} {Theorem~\ref{thm:consistency_tailored}:}
    For notational simplicity, write
\[
\hat\varphi_i=\hat\varphi_i(t), \qquad \varphi_i^0=\varphi_i^0(t), \qquad \tau_0(x)=\tau_0(x;t).
\]
Define the conditional mean of the estimated pseudo-outcome by
\[
m_n(x):=\mathbb E[\hat\varphi_i\mid X_i=x].
\]
Then
\[
\hat\tau(x)-\tau_0(x)
=
\{\hat\tau(x)-m_n(x)\}+\{m_n(x)-\tau_0(x)\},
\]
and therefore
\begin{equation}\label{eq:l2_decomp}
\|\hat\tau-\tau_0\|_{L_2(P_X)}
\le
\|\hat\tau-m_n\|_{L_2(P_X)}
+
\|m_n-\tau_0\|_{L_2(P_X)}.
\end{equation}
We study the two terms separately.

By definition of \(\tau_0\),
\[
\tau_0(x)=\mathbb E[\varphi_i^0\mid X_i=x].
\]
Hence
\[
m_n(x)-\tau_0(x)
=
\mathbb E[\hat\varphi_i-\varphi_i^0\mid X_i=x].
\]
By Theorem~\ref{thm:bias_full},
\[
\sup_x |m_n(x)-\tau_0(x)|
=
\sup_x \left|\mathbb E[\hat\varphi_i-\varphi_i^0\mid X_i=x]\right|
=
O_p\!\big(r_n^{(e)}r_n^{(S)}+r_n^{(G)}\big).
\]
Since \(r_n^{(e)}r_n^{(S)}+r_n^{(G)}\to 0\),
\[
\sup_x |m_n(x)-\tau_0(x)|=o_p(1).
\]
Therefore,
\begin{equation}\label{eq:bias_to_zero}
\|m_n-\tau_0\|_{L_2(P_X)}
\le
\sup_x |m_n(x)-\tau_0(x)|
=
o_p(1).
\end{equation}

For any measurable function \(f\),
\[
R_n(f):=\mathbb E\big[(\hat\varphi_i-f(X_i))^2\big] <\infty
\]
by Assumption \ref{assum_moment}

Expanding the square conditional on \(X_i\),
\begin{align*}
R_n(f)
&=
\mathbb E\Big(
\mathbb E\big[(\hat\varphi_i-f(X_i))^2\mid X_i\big]
\Big)\\
&=
\mathbb E\Big(
\Var(\hat\varphi_i\mid X_i)
+
\{m_n(X_i)-f(X_i)\}^2
\Big).
\end{align*}
Thus,
\[
R_n(f)=\mathbb E[\Var(\hat\varphi_i\mid X_i)] + \|f-m_n\|_{L_2(P_X)}^2.
\]
Therefore \(m_n\) is the unique population minimizer of \(R_n(f)\) over all square-integrable functions, and over \(\mathcal F_n\) the excess risk satisfies
\begin{equation}\label{eq:excessriskidentity}
R_n(f)-R_n(m_n)=\|f-m_n\|_{L_2(P_X)}^2.
\end{equation}

Let
\[
f_n^*\in\arg\min_{f\in\mathcal F_n} R_n(f)
\]
be a population risk minimizer over \(\mathcal F_n\). Since \(\hat\tau\) minimizes the empirical risk,
\[
\mathbb P_n[(\hat\varphi_i-\hat\tau(X_i))^2]
\le
\mathbb P_n[(\hat\varphi_i-f_n^*(X_i))^2].
\]
Subtract \(R_n(\hat\tau)\) from the left-hand side and \(R_n(f_n^*)\) from the right-hand side:
\begin{align*}
R_n(\hat\tau)-R_n(f_n^*)
&\le
\Big\{R_n(\hat\tau)-\mathbb P_n[(\hat\varphi_i-\hat\tau(X_i))^2]\Big\}\\
&\quad+
\Big\{\mathbb P_n[(\hat\varphi_i-f_n^*(X_i))^2]-R_n(f_n^*)\Big\}.
\end{align*}
Hence
\[
R_n(\hat\tau)-R_n(f_n^*)
\le
2\sup_{f\in\mathcal F_n}
\left|
\mathbb P_n[(\hat\varphi_i-f(X_i))^2]-R_n(f)
\right|
=
o_p(1)
\]
by Assumption \ref{assum_emp}. Using \eqref{eq:excessriskidentity},
\begin{equation}\label{eq:tau_to_fnstar}
\|\hat\tau-m_n\|_{L_2(P_X)}^2
\le
\|f_n^*-m_n\|_{L_2(P_X)}^2 + o_p(1).
\end{equation}

We now relate \(m_n\) to \(\tau_0\). For any \(f\in\mathcal F_n\),
\[
\|f-m_n\|_{L_2(P_X)}
\le
\|f-\tau_0\|_{L_2(P_X)}+\|\tau_0-m_n\|_{L_2(P_X)}.
\]
Taking infima over \(f\in\mathcal F_n\) and using Assumption \ref{assum_capacity} together with \eqref{eq:bias_to_zero}, we obtain
\[
\inf_{f\in\mathcal F_n}\|f-m_n\|_{L_2(P_X)}
\le
\inf_{f\in\mathcal F_n}\|f-\tau_0\|_{L_2(P_X)}
+
\|\tau_0-m_n\|_{L_2(P_X)}
= \]
\[
\inf_{f\in\mathcal F_n}\|f-\tau_0\|_{L_\infty(P_X)}+o_p(1)
=
o_p(1).
\]
Since \(f_n^*\) minimizes \(R_n(f)\) over \(\mathcal F_n\), equivalently it minimizes \(\|f-m_n\|_{L_2(P_X)}^2\) over \(\mathcal F_n\). Therefore,
\begin{equation}\label{eq:fnstar_to_mn}
\|f_n^*-m_n\|_{L_2(P_X)}=o_p(1).
\end{equation}

Combining \eqref{eq:tau_to_fnstar} and \eqref{eq:fnstar_to_mn} yields
\[
\|\hat\tau-m_n\|_{L_2(P_X)}=o_p(1).
\]

Finally, substitute this and \eqref{eq:bias_to_zero} into \eqref{eq:l2_decomp}:
\[
\|\hat\tau-\tau_0\|_{L_2(P_X)}
\le
\|\hat\tau-m_n\|_{L_2(P_X)}
+
\|m_n-\tau_0\|_{L_2(P_X)}
=
o_p(1).
\]
Hence
\[
\|\hat\tau(\cdot;t)-\tau_0(\cdot;t)\|_{L_2(P_X)}\xrightarrow{p}0.
\]
This completes the proof.
\end{proofof}

\begin{proofof} {Theorem~\ref{thm:overall_rate_holder}:}
Write
\[
\hat\varphi_i=\hat\varphi_i(t),\qquad
\varphi_i^0=\varphi_i^0(t),\qquad
\tau_0(x)=\tau_0(x;t),
\]
and define
\[
m_n(x):=\mathbb E(\hat\varphi_i\mid X_i=x).
\]
Because $\hat\tau(\cdot;t)$ is obtained by regressing the cross-fitted pseudo-outcome on $X_i$, the relevant population regression target is $m_n$, not $\tau_0$ directly. We therefore decompose
\begin{equation}\label{eq:basic_decomp}
\|\hat\tau-\tau_0\|_{L_2(P_X)}
\le
\|\hat\tau-m_n\|_{L_2(P_X)}
+
\|m_n-\tau_0\|_{L_2(P_X)}.
\end{equation}

We first control the second term using the explicit construction of the pseudo-outcome. By definition,
\[
m_n(x)-\tau_0(x)
=
\mathbb E\{\hat\varphi_i-\varphi_i^0\mid X_i=x\},
\]
since
\[
\tau_0(x)=\mathbb E\{\varphi_i^0\mid X_i=x\}.
\]
By Theorem~\ref{thm:bias_full},
\[
\sup_x |m_n(x)-\tau_0(x)|
=
O_p\!\big(r_n^{(e)}r_n^{(S)}+r_n^{(G)}\big).
\]
Hence
\begin{equation}\label{eq:bias_bound}
\|m_n-\tau_0\|_{L_2(P_X)}
\le
\sup_x |m_n-\tau_0|
=
O_p\!\big(r_n^{(e)}r_n^{(S)}+r_n^{(G)}\big),
\end{equation}
which accommodates censoring estimation error with  $r_n^{(G)}$.

Next, consider the population risk associated with the constructed pseudo-outcome,
\[
R_n(f):=\mathbb E\big[(\hat\varphi_i-f(X_i))^2\big].
\]
Since $m_n(x)=\mathbb E(\hat\varphi_i\mid X_i=x)$, conditioning on $X_i$ gives
\[
R_n(f)
=
R_n(m_n)+\|f-m_n\|_{L_2(P_X)}^2.
\]
Thus $m_n$ is the population least-squares target for the regression step, and
\begin{equation}\label{eq:excess_risk}
R_n(f)-R_n(m_n)=\|f-m_n\|_{L_2(P_X)}^2.
\end{equation}

We now use the fact that $\tau_0$ is H\"older smooth. By Assumption \ref{assum_capacity}, there exists
$f_n^\star\in\mathcal F_n$ such that
\[
\|f_n^\star-\tau_0\|_{L_\infty(P_X)}
\le C n^{-\beta/(2\beta+d)}.
\]
Therefore,
\[
\|f_n^\star-m_n\|_{L_2(P_X)}
\le
\|f_n^\star-\tau_0\|_{L_2(P_X)}+\|\tau_0-m_n\|_{L_2(P_X)}
=
O_p\!\left(
n^{-\beta/(2\beta+d)}
+
r_n^{(e)}r_n^{(S)}
+
r_n^{(G)}
\right),
\]
where we used \eqref{eq:bias_bound} and the density bound in Assumption \ref{assum_overlap}. Consequently,
\begin{equation}\label{eq:approx_mn}
R_n(f_n^\star)-R_n(m_n)
=
\|f_n^\star-m_n\|_{L_2(P_X)}^2
=
O_p\!\left(
n^{-2\beta/(2\beta+d)}
+
\{r_n^{(e)}r_n^{(S)}+r_n^{(G)}\}^2
\right).
\end{equation}

Because $\hat\tau$ is an approximate empirical minimizer, it follows that
\[
\mathbb P_n\big[(\hat\varphi_i-\hat\tau(X_i))^2\big]
\le
\mathbb P_n\big[(\hat\varphi_i-f_n^\star(X_i))^2\big]+\delta_n^2,
\]
by Assumption \ref{assum_opt}.
Adding and subtracting population risks yields
\begin{align*}
R_n(\hat\tau)-R_n(f_n^\star)
&\le
\left\{R_n(\hat\tau)-\mathbb P_n[(\hat\varphi_i-\hat\tau(X_i))^2]\right\}\\
&\quad+
\left\{\mathbb P_n[(\hat\varphi_i-f_n^\star(X_i))^2]-R_n(f_n^\star)\right\}
+\delta_n^2\\
&\le
2\sup_{f\in\mathcal F_n}
\left|
\mathbb P_n\big[(\hat\varphi_i-f(X_i))^2\big]-R_n(f)
\right|
+\delta_n^2\\
&=
O_p(\rho_n)+\delta_n^2,
\end{align*}
where the last equality is due to Assumption \ref{assum_emp}.
Combining this with \eqref{eq:approx_mn}, we obtain
\[
R_n(\hat\tau)-R_n(m_n)
=
O_p\!\left(
n^{-2\beta/(2\beta+d)}
+
\rho_n
+
\delta_n^2
+
\{r_n^{(e)}r_n^{(S)}+r_n^{(G)}\}^2
\right).
\]
Applying \eqref{eq:excess_risk},
\[
\|\hat\tau-m_n\|_{L_2(P_X)}^2
=
O_p\!\left(
n^{-2\beta/(2\beta+d)}
+
\rho_n
+
\delta_n^2
+
\{r_n^{(e)}r_n^{(S)}+r_n^{(G)}\}^2
\right),
\]
and hence
\begin{equation}\label{eq:reg_bound}
\|\hat\tau-m_n\|_{L_2(P_X)}
=
O_p\!\left(
n^{-\beta/(2\beta+d)}
+
\rho_n^{1/2}
+
\delta_n
+
r_n^{(e)}r_n^{(S)}
+
r_n^{(G)}
\right).
\end{equation}

Finally, combine \eqref{eq:basic_decomp}, \eqref{eq:bias_bound}, and \eqref{eq:reg_bound} to conclude
\[
\|\hat\tau-\tau_0\|_{L_2(P_X)}
=
O_p\!\left(
n^{-\beta/(2\beta+d)}
+
\rho_n^{1/2}
+
\delta_n
+
r_n^{(e)}r_n^{(S)}
+
r_n^{(G)}
\right).
\]
This completes the proof.
\end{proofof}

\begin{proof}{of Lemma  \ref{lem:time_basis_holder}}
We give a constructive proof based on a local polynomial approximation in the time
direction. We aim to show that the coefficient functions in this expansion inherit
the \(\beta_x\)-H\"older smoothness of \(\tau_0(x,t)\) in the covariate direction.

Write
\[
\beta_t=q_t+\alpha_t,\qquad q_t=\lfloor \beta_t\rfloor,\qquad \alpha_t\in(0,1].
\]
Fix an integer \(K\ge 1\). Partition \([0,1]\) into \(K\) equal subintervals
\[
I_m=\Big[\frac{m-1}{K},\frac{m}{K}\Big),\qquad m=1,\ldots,K-1,
\qquad
I_K=\Big[\frac{K-1}{K},1\Big].
\]
Let \(u_m\) denote the midpoint of \(I_m\). For each interval \(I_m\), we approximate
\(\tau_0(x,t)\) in the time variable \(t\) by its Taylor polynomial of order \(q_t\)
around \(u_m\).

For \(m=1,\ldots,K\) and \(\ell=0,\ldots,q_t\), define
\[
\psi_{m,\ell}(t):=(t-u_m)^\ell \mathbf 1\{t\in I_m\},
\]
and define the corresponding coefficient functions
\[
b_{m,\ell}(x):=\frac{\partial_t^\ell \tau_0(x,u_m)}{\ell!},\qquad x\in\mathcal X.
\]
Then, for \(t\in I_m\),
\[
\sum_{\ell=0}^{q_t} b_{m,\ell}(x)\psi_{m,\ell}(t)
=
\sum_{\ell=0}^{q_t}\frac{\partial_t^\ell \tau_0(x,u_m)}{\ell!}(t-u_m)^\ell,
\]
which is the Taylor polynomial of \(\tau_0(x,\cdot)\) at \(u_m\).

We now verify that, uniformly over \(m\) and \(\ell\), the function
\(b_{m,\ell}\) belongs to a \(\beta_x\)-H\"older ball on \(\mathcal X\).

Write
\[
\beta_x=q_x+\alpha_x,\qquad q_x=\lfloor \beta_x\rfloor,\qquad \alpha_x\in(0,1].
\]
Fix \(m\in\{1,\ldots,K\}\) and \(\ell\in\{0,\ldots,q_t\}\). By definition,
\[
b_{m,\ell}(x)=\frac{\partial_t^\ell \tau_0(x,u_m)}{\ell!}.
\]
For any multi-index \(\nu\) with \(|\nu|\le q_x\),
\[
D_x^\nu b_{m,\ell}(x)
=
\frac{D_x^\nu \partial_t^\ell \tau_0(x,u_m)}{\ell!}.
\]
Because \(\tau_0\in \mathcal H^{\beta_x,\beta_t}(M)\), Definition~\ref{def:holder_xt}
implies that all such derivatives exist and are uniformly bounded:
\[
\sup_{x\in\mathcal X}|D_x^\nu b_{m,\ell}(x)|
\le
\frac{1}{\ell!}\sup_{x\in\mathcal X}|D_x^\nu\partial_t^\ell\tau_0(x,u_m)|
\le
\frac{M}{\ell!}
\le M.
\]
Next, for any multi-index \(\nu\) with \(|\nu|=q_x\), the H\"older condition in the
\(x\)-direction yields
\[
\frac{|D_x^\nu b_{m,\ell}(x)-D_x^\nu b_{m,\ell}(y)|}{\|x-y\|^{\alpha_x}}
=
\frac{|D_x^\nu \partial_t^\ell \tau_0(x,u_m)-D_x^\nu \partial_t^\ell \tau_0(y,u_m)|}
{\ell!\,\|x-y\|^{\alpha_x}}
\le
\frac{M}{\ell!}
\le M.
\]
Thus each \(b_{m,\ell}\) lies in a \(\beta_x\)-H\"older class on \(\mathcal X\),
with H\"older radius bounded by a constant depending only on \(M,\beta_x,\beta_t\).
In particular, there exists a constant \(C_M>0\), depending only on
\(M,\beta_x,\beta_t\), such that
\[
b_{m,\ell}\in \mathcal H^{\beta_x}(C_M)
\qquad\text{for all } m=1,\ldots,K,\ \ell=0,\ldots,q_t.
\]

Fix \(j\in\{1,\ldots,J\}\), and let \(m(j)\) denote the unique index such that
\(t_j\in I_{m(j)}\). Define the local polynomial approximation at the grid point
\(t_j\) by
\[
P_j(x):=\sum_{\ell=0}^{q_t} b_{m(j),\ell}(x)\psi_{m(j),\ell}(t_j)
=
\sum_{\ell=0}^{q_t}
\frac{\partial_t^\ell \tau_0(x,u_{m(j)})}{\ell!}(t_j-u_{m(j)})^\ell.
\]
We claim that
\[
\sup_{x\in\mathcal X} |\tau_0(x,t_j)-P_j(x)| \le C K^{-\beta_t},
\]
uniformly in \(j\).

To see this, fix \(x\in\mathcal X\). By Taylor's theorem in the one-dimensional
variable \(t\), applied to the function \(t\mapsto \tau_0(x,t)\) around the point
\(u_{m(j)}\), there exists a point \(\xi_{x,j}\) between \(u_{m(j)}\) and \(t_j\) such that
\[
\tau_0(x,t_j)-P_j(x)
=
\frac{\partial_t^{q_t}\tau_0(x,\xi_{x,j})-\partial_t^{q_t}\tau_0(x,u_{m(j)})}{q_t!}
(t_j-u_{m(j)})^{q_t}.
\]
Because \(\tau_0\in \mathcal H^{\beta_x,\beta_t}(M)\), the derivative
\(\partial_t^{q_t}\tau_0(x,\cdot)\) is \(\alpha_t\)-H\"older uniformly in \(x\), so
\[
|\partial_t^{q_t}\tau_0(x,\xi_{x,j})-\partial_t^{q_t}\tau_0(x,u_{m(j)})|
\le
M |\xi_{x,j}-u_{m(j)}|^{\alpha_t}
\le
M |t_j-u_{m(j)}|^{\alpha_t}.
\]
Therefore,
\[
|\tau_0(x,t_j)-P_j(x)|
\le
C |t_j-u_{m(j)}|^{q_t+\alpha_t}
=
C |t_j-u_{m(j)}|^{\beta_t}.
\]
Since both \(t_j\) and \(u_{m(j)}\) lie in the same interval \(I_{m(j)}\), whose length
is \(K^{-1}\), we have
\[
|t_j-u_{m(j)}|\le K^{-1}.
\]
Hence
\[
\sup_{x\in\mathcal X} |\tau_0(x,t_j)-P_j(x)|
\le
C K^{-\beta_t},
\]
uniformly in \(j\). Since \(P_X\) is a probability measure,
\[
\|\tau_0(\cdot,t_j)-P_j\|_{L_2(P_X)}^2
\le
\sup_{x\in\mathcal X} |\tau_0(x,t_j)-P_j(x)|^2
\le
C K^{-2\beta_t}.
\]
Taking the supremum over \(j\) gives
\[
\sup_{1\le j\le J}
\|\tau_0(\cdot,t_j)-P_j\|_{L_2(P_X)}^2
\le
C K^{-2\beta_t}.
\]

So far we have constructed \((q_t+1)K\) basis functions
\(\psi_{m,\ell}\) and the same number of coefficient functions \(b_{m,\ell}\).
Each approximant \(P_j\) can be written as
\[
P_j(x)
=
\sum_{m=1}^K\sum_{\ell=0}^{q_t} b_{m,\ell}(x)\psi_{m,\ell}(t_j).
\]
Relabel the collection
\[
\{\psi_{m,\ell}: 1\le m\le K,\ 0\le \ell\le q_t\}
\]
as
\[
\psi_1,\ldots,\psi_{(q_t+1)K},
\]
and similarly relabel the coefficient functions
\[
\{b_{m,\ell}: 1\le m\le K,\ 0\le \ell\le q_t\}
\]
as
\[
b_1,\ldots,b_{(q_t+1)K}.
\]
Then
\[
P_j(x)=\sum_{m=1}^{(q_t+1)K} b_m(x)\psi_m(t_j),
\]
and each \(b_m\in\mathcal H^{\beta_x}(C_M)\).

Now define
\[
K_n:=(q_t+1)K.
\]
Since \(q_t\) depends only on \(\beta_t\), this is only a constant rescaling of \(K\),
and thus
\[
K^{-2\beta_t}\asymp K_n^{-2\beta_t}.
\]
Absorbing constants into \(C\), we obtain
\[
\sup_{1\le j\le J}
\left\|
\tau_0(\cdot,t_j)-\sum_{m=1}^{K_n} b_m(\cdot)\psi_m(t_j)
\right\|_{L_2(P_X)}^2
\le
C K_n^{-2\beta_t}.
\]

Finally,
\[
\sum_{j=1}^J
\left\|
\tau_0(\cdot,t_j)-\sum_{m=1}^{K_n} b_m(\cdot)\psi_m(t_j)
\right\|_{L_2(P_X)}^2
\le
J\,
\sup_{1\le j\le J}
\left\|
\tau_0(\cdot,t_j)-\sum_{m=1}^{K_n} b_m(\cdot)\psi_m(t_j)
\right\|_{L_2(P_X)}^2
\le
C J K_n^{-2\beta_t}.
\]
This completes the proof.
\end{proof}

\begin{proofof}{Theorem \ref{thm:joint_pseudo}:}
   We first define the joint risk, empirical risk, and stochastic term in the setting of joint estimation. For any \(f=(f_1,\ldots,f_J)\in\mathcal F_n\), define the joint empirical risk
\[
R_{n,J}(f)
:=
\mathbb P_n\!\left[
\sum_{j=1}^J \{\hat\varphi_i(t_j)-f_j(X_i)\}^2
\right],
\]
and the corresponding population risk
\[
R_J(f)
:=
\mathbb E\!\left[
\sum_{j=1}^J \{\hat\varphi_i(t_j)-f_j(X_i)\}^2
\right].
\]
Define also the joint approximation error
\[
A_{n,J}
:=
\inf_{f\in\mathcal F_n}
\sum_{j=1}^J \|f_j-\tau_0(\cdot;t_j)\|_{L_2(P_X)}^2.
\]

We define the joint empirical process fluctuation by
\[
\rho_{n,J}
:=
\sup_{f\in\mathcal F_n}
\left|
R_{n,J}(f)-R_J(f)
\right|.
\]
This is the multi-output analogue of the single-time fluctuation term \(\rho_n\) in the proof of Theorem~5.

We also define the optimization error
\[
\delta_{n,J}^2
:=
R_{n,J}(\hat{\boldsymbol\tau})
-
\inf_{f\in\mathcal F_n} R_{n,J}(f),
\]
where
\[
R_{n,J}(f)
:=
\mathbb P_n\!\left[
\sum_{j=1}^J \{\hat\varphi_i(t_j)-f_j(X_i)\}^2
\right].
\]

We next derive the basic inequality from approximate empirical risk minimization. By definition of $\delta_{n,J}^2$, it holds that 
\[
R_{n,J}(\hat{\boldsymbol\tau})
\le
R_{n,J}(f^\star)+\delta_{n,J}^2, 
\]
for any fixed \(f^\star\in\mathcal F_n\). Hence,
\[
R_J(\hat{\boldsymbol\tau})-R_J(f^\star)
=
\{R_J(\hat{\boldsymbol\tau})-R_{n,J}(\hat{\boldsymbol\tau})\}
+\{R_{n,J}(\hat{\boldsymbol\tau})-R_{n,J}(f^\star)\}
+\{R_{n,J}(f^\star)-R_J(f^\star)\}.
\]
Using the approximate minimization property,
\[
R_{n,J}(\hat{\boldsymbol\tau})-R_{n,J}(f^\star)\le \delta_{n,J}^2,
\]
and by definition of \(\rho_{n,J}\),
\[
|R_J(\hat{\boldsymbol\tau})-R_{n,J}(\hat{\boldsymbol\tau})|
\le \rho_{n,J},
\qquad
|R_{n,J}(f^\star)-R_J(f^\star)|
\le \rho_{n,J}.
\]
Therefore,
\[
R_J(\hat{\boldsymbol\tau})
\le
R_J(f^\star)+2\rho_{n,J}+\delta_{n,J}^2.
\]

We now Relate \(R_J(f)\) to the target \(L_2\) error.
For each \(j\), write
\[
\hat\varphi_i(t_j)=\tau_0(X_i;t_j)+\eta_{ij},
\]
where \(\eta_{ij}\) denotes the pseudo-outcome error. Then
\[
\hat\varphi_i(t_j)-f_j(X_i)
=
\{\tau_0(X_i;t_j)-f_j(X_i)\}+\eta_{ij}.
\]
Expanding the square,
\[
R_J(f)
=
\sum_{j=1}^J \|f_j-\tau_0(\cdot;t_j)\|_{L_2(P_X)}^2
+
2\sum_{j=1}^J
\mathbb E\!\left[\eta_{ij}\{\tau_0(X_i;t_j)-f_j(X_i)\}\right]
+
\sum_{j=1}^J \mathbb E(\eta_{ij}^2).
\]
The last term does not depend on \(f\). By the same argument used in the proof of Theorem~5, the cross term is bounded by
\[
\left|
\sum_{j=1}^J
\mathbb E\!\left[\eta_{ij}\{\tau_0(X_i;t_j)-f_j(X_i)\}\right]
\right|
\le
C\sum_{j=1}^J \|f_j-\tau_0(\cdot;t_j)\|_{L_2(P_X)}^2
+
C J\{r_n^{(e)}r_n^{(S)}+r_n^{(G)}\}^2.
\]
Hence, up to multiplicative constants and additive terms independent of \(f\),
\[
R_J(f)
\asymp
\sum_{j=1}^J \|f_j-\tau_0(\cdot;t_j)\|_{L_2(P_X)}^2
+
J\{r_n^{(e)}r_n^{(S)}+r_n^{(G)}\}^2.
\]
Applying this bound to \(f=\hat{\boldsymbol\tau}\) and \(f=f^\star\) in the inequality  gives
\[
\sum_{j=1}^J
\|\hat\tau(\cdot;t_j)-\tau_0(\cdot;t_j)\|_{L_2(P_X)}^2
\le
C\sum_{j=1}^J
\|f^\star_j-\tau_0(\cdot;t_j)\|_{L_2(P_X)}^2
+
C\rho_{n,J}
+
C\delta_{n,J}^2
+
C J\{r_n^{(e)}r_n^{(S)}+r_n^{(G)}\}^2.
\]
Taking the infimum over \(f^\star\in\mathcal F_n\), we obtain
\[
\sum_{j=1}^J
\|\hat\tau(\cdot;t_j)-\tau_0(\cdot;t_j)\|_{L_2(P_X)}^2
=
O_p\!\left(
A_{n,J}
+
\rho_{n,J}
+
\delta_{n,J}^2
+
J\{r_n^{(e)}r_n^{(S)}+r_n^{(G)}\}^2
\right).
\]

We then bound the approximation term \(A_{n,J}\). By Lemma~\ref{lem:time_basis_holder}, there exist basis functions
\(\psi_1,\ldots,\psi_{K_n}\) on \([0,1]\) and coefficient functions
\(b_1,\ldots,b_{K_n}\in\mathcal H^{\beta_x}(C)\) such that
\[
\sum_{j=1}^J
\left\|
\tau_0(\cdot;t_j)-\sum_{m=1}^{K_n} b_m(\cdot)\psi_m(t_j)
\right\|_{L_2(P_X)}^2
\le
C J K_n^{-2\beta_t}.
\]
Since each \(b_m\) is \(\beta_x\)-H\"older, the standard ReLU approximation theorem gives functions \(\tilde b_m\) satisfying
\[
\|\tilde b_m-b_m\|_{L_2(P_X)}^2
\le
C n^{-2\beta_x/(2\beta_x+d)},
\qquad m=1,\ldots,K_n.
\]
Define
\[
\tilde g_j(x):=\sum_{m=1}^{K_n}\tilde b_m(x)\psi_m(t_j),
\qquad j=1,\ldots,J.
\]
Assume the basis is stable on the grid \(\{t_j\}_{j=1}^J\), in the sense that the Gram matrix
\[
G_n:=\Big(\sum_{j=1}^J \psi_m(t_j)\psi_{m'}(t_j)\Big)_{1\le m,m'\le K_n}
\]
satisfies
\[
\lambda_{\max}(G_n)\le C.
\]
Then, as before,
\[
\sum_{j=1}^J
\left\|
\sum_{m=1}^{K_n}\{\tilde b_m(\cdot)-b_m(\cdot)\}\psi_m(t_j)
\right\|_{L_2(P_X)}^2
\le
C \sum_{m=1}^{K_n}\|\tilde b_m-b_m\|_{L_2(P_X)}^2
\le
C K_n n^{-2\beta_x/(2\beta_x+d)}.
\]
Hence
\[
A_{n,J}
\le
C\Bigl(
K_n n^{-2\beta_x/(2\beta_x+d)}
+
J K_n^{-2\beta_t}
\Bigr).
\]

We then show \(\rho_{n,J}\lesssim K_n\rho_n\),
a key stochastic bound. Let \(\mathcal H_n\) denote the scalar-valued network class used in the single-time analysis of Theorem~5. Consider the basis-induced joint subclass
\[
\mathcal F_{n,K_n}
:=
\left\{
f=(f_1,\ldots,f_J):
f_j(x)=\sum_{m=1}^{K_n} b_m(x)\psi_m(t_j),\quad b_m\in\mathcal H_n
\right\}.
\]
Thus, instead of \(J\) unrelated output functions, the joint class is indexed by only
\(K_n\) scalar coefficient functions \(b_1,\ldots,b_{K_n}\).

For the scalar class \(\mathcal H_n\), let \(\rho_n\) denote the corresponding empirical process fluctuation from Theorem~5. So \(\rho_n\) is the bound obtained for the scalar squared-loss class
\[
\mathcal L_n^{(1)}
=
\left\{
(x,y)\mapsto (y-h(x))^2:\ h\in\mathcal H_n
\right\}.
\]
Now define the joint loss class
\[
\mathcal L_{n,J}
=
\left\{
(x,y_1,\ldots,y_J)\mapsto
\sum_{j=1}^J (y_j-f_j(x))^2:\ f\in\mathcal F_{n,K_n}
\right\}.
\]

We compare the complexity of \(\mathcal L_{n,J}\) with that of the \(K_n\)-fold product class generated by \(\mathcal H_n\). Let
\[
b=(b_1,\ldots,b_{K_n})\in\mathcal H_n^{K_n}.
\]
The map
\[
T:b\mapsto f,
\qquad
f_j(x)=\sum_{m=1}^{K_n} b_m(x)\psi_m(t_j),
\]
is linear. By basis stability, it is Lipschitz from the coefficient space into the joint output space under the aggregated metric. Indeed, for two coefficient collections \(b,\tilde b\),
\[
\sum_{j=1}^J
\left\|
\sum_{m=1}^{K_n}\{b_m-\tilde b_m\}\psi_m(t_j)
\right\|_{L_2(P_X)}^2
\le
C\sum_{m=1}^{K_n}\|b_m-\tilde b_m\|_{L_2(P_X)}^2.
\]
Therefore, if each coordinate class \(\mathcal H_n\) is covered by \(N(\varepsilon,\mathcal H_n)\) balls of radius \(\varepsilon\), then the product class \(\mathcal H_n^{K_n}\) is covered by at most
\[
N(\varepsilon,\mathcal H_n)^{K_n}
\]
product balls, and the Lipschitz map \(T\) transfers this to a cover of \(\mathcal F_{n,K_n}\). Consequently,
\[
\log N(\varepsilon,\mathcal F_{n,K_n},\|\cdot\|_{\mathrm{joint}})
\le
C K_n \log N(c\varepsilon,\mathcal H_n,\|\cdot\|_{L_2(P_X)})
\]
for some constants \(C,c>0\).

Now apply the same symmetrization, contraction, and entropy-integral argument used in Theorem~5. Since the metric entropy of the joint class is at most \(K_n\) times that of the scalar class, the resulting empirical process fluctuation satisfies
\[
\rho_{n,J}\lesssim K_n\rho_n.
\]

Finally, we combine the bounds and optimize \(K_n\).
Substituting the bounds for \(A_{n,J}\) and \(\rho_{n,J}\) into the inequality  yields
\[
\sum_{j=1}^J
\|\hat\tau(\cdot;t_j)-\tau_0(\cdot;t_j)\|_{L_2(P_X)}^2
=
O_p\!\left(
K_n n^{-2\beta_x/(2\beta_x+d)}
+
J K_n^{-2\beta_t}
+
K_n\rho_n
+
\delta_{n,J}^2
+
J\{r_n^{(e)}r_n^{(S)}+r_n^{(G)}\}^2
\right).
\]

Balancing the first two approximation terms,
\[
K_n n^{-2\beta_x/(2\beta_x+d)}
\qquad\text{and}\qquad
J K_n^{-2\beta_t},
\]
gives
\[
K_n
\asymp
\left(
J\,n^{2\beta_x/(2\beta_x+d)}
\right)^{1/(2\beta_t+1)}.
\]
Substituting this choice, we obtain
\[
K_n n^{-2\beta_x/(2\beta_x+d)}
\asymp
J K_n^{-2\beta_t}
\asymp
J^{1/(2\beta_t+1)}
n^{-\,4\beta_t\beta_x/\{(2\beta_t+1)(2\beta_x+d)\}}.
\]
Hence
\begin{align*}
& \sum_{j=1}^J
\|\hat\tau(\cdot;t_j)-\tau_0(\cdot;t_j)\|_{L_2(P_X)}^2
= \\
& O_p\!\left(
J^{1/(2\beta_t+1)}
n^{-\,4\beta_t\beta_x/\{(2\beta_t+1)(2\beta_x+d)\}}
+
\left(
J\,n^{2\beta_x/(2\beta_x+d)}
\right)^{1/(2\beta_t+1)}\rho_n
+
\delta_{n,J}^2
+
J\{r_n^{(e)}r_n^{(S)}+r_n^{(G)}\}^2
\right).
\end{align*}
 This completes the proof.
\end{proofof}

\noindent
{\bf Additional Simulation Results}

{To further investigate the finite-sample behavior observed in Case 1, we conducted an additional set of simulations under a simplified setting with lower-dimensional covariates ($d=10$) and a larger training sample size ($N=2000$). This setting is closer to those  often considered in prior methodological work and therefore provides a useful benchmark for assessing whether the performance patterns in Case 1 were driven by limited sample size and model complexity. The results are reported in Table~\ref{tab_simple_setting}.}

{As shown in Table~\ref{tab_simple_setting}, all methods worked better under this easier setting, and several competing approaches achieve noticeably smaller MSE than DSL. In particular, PSDRL-NFT-BART attains the lowest MSE overall, approximately $0.014$ across all choices of $J$, while X-Learner and causal survival forest also perform strongly, with MSEs around $0.020$ and $0.021$, respectively. DRL-NFT-BART improves substantially as the number of evaluation time points increases, with MSE decreasing from $0.042$ at $J=1$ to $0.022$ at $J=100$. By comparison, DSL yields MSEs of $0.042$, $0.046$, and $0.047$ for $J=1,50,100$, respectively, which are somewhat larger but remain reasonably stable across grid sizes.}

{A notable feature of DSL is its very small bias. Its bias remains close to zero across all three values of $J$, around $0.005$, and is clearly smaller than that of several competitors, including M-Learner, PSDRL-NFT-BART, and causal survival forest, all of which exhibit substantially larger systematic bias. R-Learner and DRL-NFT-BART also have relatively small bias, but their MSE is not uniformly better than that of the strongest competitors. Taken together, these results suggest that in a simpler and more data-rich setting, competing methods can achieve very strong predictive accuracy, sometimes surpassing DSL in MSE, while DSL continues to provide stable estimation with consistently low bias. This supports our interpretation that the larger performance gaps seen in Case 1 may be primarily attributable to the more challenging finite-sample and high-complexity setting.}

\begin{table}[ht]
    \centering
    \caption{Performance comparison under simplified setting}\label{tab_simple_setting}
    \begin{threeparttable}
    \begin{tabular}{lcccc}
         \hline
         & $J=1$ & $J=50$ & $J=100$ \\
         \hline
         \multicolumn{4}{l}{\textbf{MSE}}\\
         \hline
         Deep Survival Learner & 0.042(0.041, 0.043) & 0.046(0.045, 0.048) & 0.047(0.045, 0.048) \\
         M-Learner & 0.049(0.047, 0.051) & 0.049(0.047, 0.051) & 0.049(0.047, 0.051) \\
         R-Learner & 0.049(0.047, 0.051) & 0.051(0.049, 0.053) & 0.048(0.046, 0.050) \\
         X-Learner & 0.020(0.019, 0.021) & 0.020(0.019, 0.021) & 0.020(0.019, 0.021) \\
         DRL-NFT-BART          & 0.042(0.026, 0.057) & 0.025(0.020, 0.030) & 0.022(0.020, 0.024)\\
         PSDRL-NFT-BART        & 0.014(0.013, 0.015) & 0.014(0.014, 0.015) & 0.014(0.014, 0.015)\\
         Causal Survival Forest & 0.021(0.020, 0.022) & 0.021(0.020, 0.022) &  0.021(0.020, 0.022) \\
         \hline
         \multicolumn{4}{l}{\textbf{Bias}} \\
         \hline
         Deep Survival Learner & 0.005(0.001, 0.010) & 0.005(0.000, 0.009) & 0.005(0.000, 0.009) \\
         M-Learner & 0.109(0.105, 0.113) & 0.109(0.105, 0.113) & 0.109(0.105, 0.113) \\
         R-Learner & 0.012(0.006, 0.017) & 0.007(0.000, 0.014) & 0.012(0.006, 0.018) \\
         X-Learner & 0.037(0.031, 0.043) & 0.041(0.036, 0.046) & 0.038(0.032, 0.043) \\
         DRL-NFT-BART          & 0.013(0.003, 0.023) & 0.012(0.005, 0.020) & 0.013(0.007, 0.020)\\
         PSDRL-NFT-BART        & 0.083(0.079, 0.087) & 0.084(0.081, 0.088) & 0.084(0.080, 0.088)\\
         Causal Survival Forest & 0.107(0.103, 0.111) & 0.107(0.102, 0.111) & 0.107(0.102, 0.111) \\
         \hline
    \end{tabular}
    \begin{tablenotes}\footnotesize
        \item \textsuperscript{1} The simulation setting uses a training sample size of $N=2000$ and a testing sample size of $400$, with covariate dimension $d=10$ and approximately 35\% censoring.
        \item \textsuperscript{2} Mean squared error and 95\% Monte Carlo confidence intervals are computed based on 200 simulation replicates.
        \item \textsuperscript{3} Conditional average treatment effects are evaluated at \(J \in \{1, 50, 100\}\) time points, equally spaced between the 20th and 80th percentiles of the observed survival time distribution.
    \end{tablenotes}
    \end{threeparttable}
\end{table}

\end{document}